\documentclass[9pt]{elife} 
\usepackage{bm}
\usepackage{multicol}
\usepackage{multirow}
\usepackage{wrapfig}

\usepackage[utf8]{inputenc} 
\usepackage[T1]{fontenc}    
\usepackage{hyperref}       
\usepackage{url}            
\usepackage{booktabs}       
\usepackage{amsmath}
\usepackage{nicefrac}       
\usepackage{microtype}      
\usepackage{hhline}
\usepackage{graphicx}
\usepackage[font=small,labelfont=bf,tableposition=top]{caption}
\DeclareCaptionLabelFormat{andtable}{#1~#2  \&  \tablename~\thetable}
\usepackage[english]{babel}
\usepackage[english]{babel}
\usepackage{float}
\usepackage{wrapfig}
\usepackage{amsthm}

\newtheorem{lemma}{Lemma}

\usepackage{listings}
\usepackage{xcolor}
\usepackage{courier}
\definecolor{codegreen}{rgb}{0,0.6,0}
\definecolor{codegray}{rgb}{0.5,0.5,0.5}
\definecolor{codepurple}{rgb}{0.58,0,0.82}
\definecolor{backcolour}{rgb}{0.95,0.95,0.92}

\lstdefinestyle{mystyle}{
    float=tp,
    floatplacement=tbp,
    commentstyle=\color{codegreen},
    keywordstyle=\color{magenta},
    numberstyle=\tiny\color{codegray},
    stringstyle=\color{codepurple},
    basicstyle=\ttfamily\small,
    breakatwhitespace=false,         
    breaklines=true,                 
    captionpos=b,                    
    keepspaces=true,                 
    numbers=left,                    
    numbersep=5pt,                  
    showspaces=false,                
    showstringspaces=false,
    showtabs=false,                  
    tabsize=2,
    frame=tb,
}

\newfloat{lstfloat}{t}{lop}
\floatname{lstfloat}{Listing}

\title{End-to-End differentiable construction of molecular mechanics force fields}

\author[1, 3, 4]{Yuanqing~Wang~(ORCID:~\href{https://orcid.org/0000-0003-4403-2015}{0000-0003-4403-2015})}
\author[1, 2\authfn{2}]{Josh~Fass~(ORCID:~\href{http://orcid.org/0000-0003-3719-266X}{0000-0003-3719-266X})}
\author[1, 2]{Benjamin~Kaminow~(ORCID:~\href{http://orcid.org/0000-0002-2266-3353}{0000-0002-2266-3353})}
\author[1]{John~E.~Herr~(ORCID:~\href{http://orcid.org/0000-0001-5639-1520}{0000-0001-5639-1520})}
\author[1]{Dominic Rufa~(ORCID:~\href{https://orcid.org/0000-0003-0930-9445}{0000-0003-0930-9445})}
\author[1, 2]{Ivy Zhang~(ORCID:~\href{https://orcid.org/0000-0003-0628-6276}{0000-0003-0628-6276})}
\author[1]{Iván Pulido~(ORCID:~\href{https://orcid.org/0000-0002-7178-8136}{0000-0002-7178-8136})}
\author[1]{Mike Henry~(ORCID:~\href{https://orcid.org/0000-0002-3870-9993}{0000-0002-3870-9993})}
\author[1]{John~D.~Chodera~(ORCID:~\href{http://orcid.org/0000-0003-0542-119X}{0000-0003-0542-119X})}

\affil[1]{Computational and Systems Biology Program, Sloan Kettering Institute, Memorial Sloan Kettering Cancer Center, New York, NY 10065}
\affil[2]{Tri-Institutional PhD Program in Computational Biology and Medicine, Weill Cornell Medical College, Cornell University,
New York, NY 10065}
\affil[3]{Physiology, Biophysics, and System Biology Ph.D.\ Program, Weill Cornell Medical College, Cornell University, New York, NY 10065}
\affil[4]{MFA.\ Program in Creative Writing, Division of Humanities and Arts, the City College of New York, the City University of New York, New York, NY 10031}

\corr{yuanqing.wang@choderalab.org}{YW}
\corr{john.chodera@choderalab.org}{JDC}

\presentadd[\authfn{3}]{Relay Therapeutics, Cambridge, MA 02139}

\begin{document}

\maketitle


\begin{abstract}
    Molecular mechanics (MM) potentials have long been a workhorse of computational chemistry.
    Leveraging accuracy and speed, these functional forms find use in a wide variety of applications in biomolecular modeling and drug discovery, from rapid virtual screening to detailed free energy calculations.
    Traditionally, MM potentials have relied on human-curated, inflexible, and poorly extensible discrete chemical perception rules (\textit{atom types}) for applying parameters to small molecules or biopolymers, making it difficult to optimize both types and parameters to fit quantum chemical or physical property data.
    Here, we propose an alternative approach that uses \textit{graph neural networks} to perceive chemical environments, producing continuous atom embeddings from which valence and nonbonded parameters can be predicted using invariance-preserving layers.
    Since all stages are built from smooth neural functions, the entire process---spanning chemical perception to parameter assignment---is modular and end-to-end differentiable with respect to model parameters, allowing new force fields to be easily constructed, extended, and applied to arbitrary molecules.
    We show that this approach is not only sufficiently expressive to reproduce legacy atom types, but that it can learn to accurately reproduce and extend existing molecular mechanics force fields.
    Trained with arbitrary loss functions, it can construct entirely new force fields self-consistently applicable to both biopolymers and small molecules directly from quantum chemical calculations, with superior fidelity than traditional atom or parameter typing schemes.
    When adapted to simultaneously fit partial charge models, espaloma delivers high-quality partial atomic charges orders of magnitude faster than current best-practices with little inaccuracy.
    When trained on the same quantum chemical small molecule dataset used to parameterize the Open Force Field ("Parsley") \textsc{openff-1.2.0} small molecule force field augmented with a peptide dataset, the resulting espaloma model shows superior accuracy vis-à-vis experiments in computing relative alchemical free energy calculations for a popular benchmark set.
    This approach is implemented in the free and open source package {\bf espaloma}, available at \url{https://github.com/choderalab/espaloma}.
\end{abstract}


Molecular mechanics (MM) force fields---physical models that abstract molecular systems as atomic point masses that interact via nonbonded interactions and valence (bond, angle, torsion) terms---have powered \textit{in silico} modeling to provide key insights and quantitative predictions in all aspects of chemistry, from drug discovery to materials science~\cite{ponder2003force, van2005gromacs, case2005amber, phillips2005scalable, calculations2007theory, WANG2014979, li2015molecular, sun2016compass, harder2016opls3}.
While recent work in quantum machine learning (QML) potentials has demonstrated how flexibility in functional forms and training strategies can lead to increased accuracy~\cite{smith2017ani,smith2018less,smith2019approaching,devereux2020extending,schutt2018schnet,batzner2021se,10.1093/bib/bbab158}, these QML potentials are orders of magnitude slower than popular molecular mechanics potentials even on expensive hardware accelerators, as they involve orders of magnitude more floating point operations per energy or force evaluation.

On the other hand, the simpler physical energy functions of MM models are compatible with highly optimized implementations that can exploit a wide variety of hardware~\cite{eastman2017openmm,van2005gromacs,harvey2009acemd,salomon2013routine,schoenholz2019jax,wang2020differentiable}, but rely on complex and inextensible legacy \textit{atom typing schemes} for parameter assignment~\cite{mobley2018escaping}:
\begin{itemize}
\item First, a set of rules is used to classify atoms into discrete \textit{atom types} that must encode all information about an atom's chemical environment needed in subsequent parameter assignment steps.
\item Next, a discrete set of bond, angle, and torsion types is determined by composing the atom types involved in the interaction.
\item Finally, the parameters attached to atoms, bonds, angles, and torsions are assigned according to a look-up table of composed parameter types.
\end{itemize}

Consequently, atoms, bonds, angles, or torsions with distinct chemical environments that happen to fall into the same expert-derived discrete type class are forced to share the same set of MM parameters, potentially leading to low resolution and potentially poor accuracy.
On the other hand, the explosion of number of discrete parameter classes describing equivalent chemical environments required by traditional MM atom typing schemes not only poses significant challenges to extending the space of atom types~\cite{mobley2018escaping}, optimizing these independently has the potential to compromise generalizabilty and lead to overfitting.
Even with modern MM parameter optimization frameworks~\cite{Wang2012, Wang2014, Qiu2019} and sufficient data, parameter optimization is only feasible in the continuous parameter space defined by these fixed atom types, while the mixed discrete-continuous optimization problem---jointly optimizing types and parameters---is intractable.


\begin{figure}[tb]
    \centering
    \includegraphics[width=0.9\textwidth]{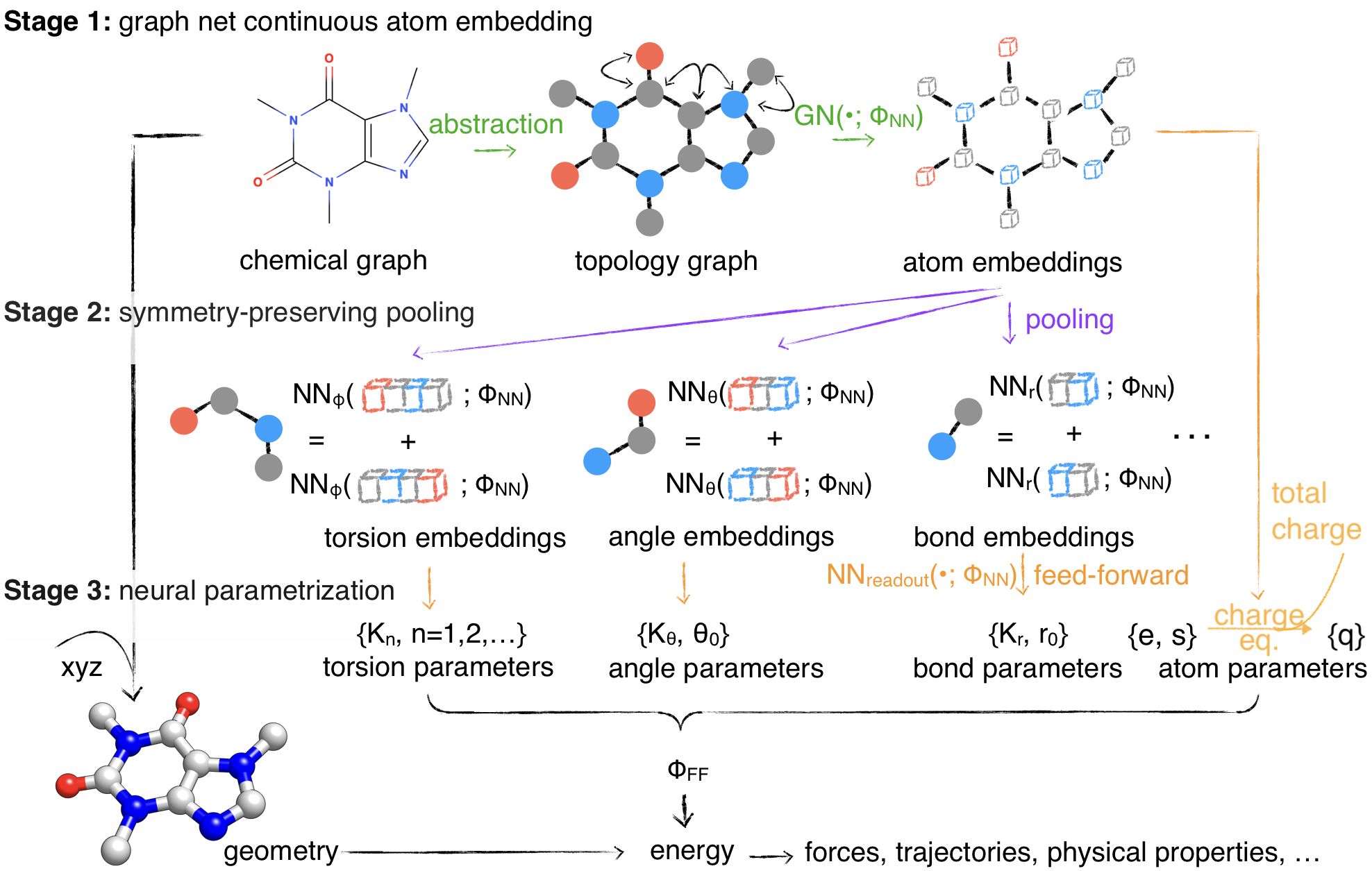}
    \caption{\label{fig:overview}
    \textbf{Espaloma is an end-to-end differentiable molecular mechanics parameter assignment scheme for arbitrary organic molecules.} 
    espaloma (\textit{\textbf{e}xtendable \textbf{s}urrogate \textbf{p}otenti\textbf{al} \textbf{o}ptimized by \textbf{m}ess\textbf{a}ge-passing}) is a modular approach for directly computing molecular mechanics force field parameters $\Phi_\mathrm{FF}$ from a chemical graph $\mathcal{G}$ such as a small molecule or biopolymer via a process that is fully differentiable in the model parameters $\Phi_\mathrm{NN}$.
    In {\bf Stage 1}, a graph neural network is used to generate continuous latent atom embeddings describing local chemical environments from the chemical graph.
    In {\bf Stage 2}, these atom embeddings are transformed into feature vectors that preserve appropriate symmetries for atom, bond, angle, and proper/improper torsion inference via Janossy pooling.
    In {\bf Stage 3}, molecular mechanics parameters are directly predicted from these feature vectors using feed-forward neural nets.
    This parameter assignment process is performed once per molecular species, allowing the potential energy to be rapidly computed using standard molecular mechanics or molecular dynamics frameworks thereafter.
    The collection of parameters $\Phi_\mathrm{NN}$ describing the espaloma model can be considered as the equivalent complete specification of a traditional molecular mechanics force field such as GAFF~\cite{wang2004development,wang2006automatic}/AM1-BCC~\cite{jakalian2000fast,jakalian2002fast} in that it encodes the equivalent of traditional typing rules, parameter assignment tables, and even partial charge models.
    This final stage is modular, and can be easily extended to incorporate additional molecular mechanics parameter classes, such as parameters for a charge-equilibration model (Section~\ref{sec:charge-equilibration}), point polarizabilities, or valence-coupling terms for Class~II molecular mechanics force fields~\cite{maple1994derivation,hwang1994derivation}.
    }
\end{figure}

Here, we present a continuous alternative to traditional discrete atom typing schemes that permits full \emph{end-to-end differentiable optimization} of both typing and parameter assignment stages, allowing an entire force field to be built, extended, and applied using standard automatic differentiation machine learning frameworks such as PyTorch~\cite{NEURIPS2019_9015}, TensorFlow~\cite{tensorflow2015-whitepaper}, or JAX~\cite{jax2018github} (\FIG{overview}).
Graph neural networks have recently emerged as a powerful way to learn chemical properties of atoms, bonds, and molecules for biomolecular species (both small organic molecules and biopolymers), which can be considered isomorphic with their graph representations~\cite{xu2018powerful, wang2019graph, DBLP:journals/corr/KipfW16, battaglia2018relational, Du:ArXiv171010370CsStat:2018, wu2019simplifying, wang2019deep, wang2019dynamic, duvenaud2015convolutional, gilmer2017neural}.
We hypothesize that graph neural networks operating on molecules have expressiveness that is at least equivalent to---and likely much greater than---expert-derived typing rules, with the advantage of being able to smoothly interpolate between representations of chemical environments (such as accounting for fractional bond orders~\cite{stern2020capturing}).
We provide empirical evidence for this in Section~\ref{sec:stage1}.

Next, we demonstrate the utility of such a model (which we call the \textit{\textbf{e}xtendable \textbf{s}urrogate \textbf{p}otenti\textbf{al} \textbf{o}ptimized by \textbf{m}essage-p\textbf{a}ssing}, or \textbf{espaloma}) to construct end-to-end optimizable force fields with continuous atom types.
Espaloma assigns molecular mechanics parameters from a molecular graph in three differentiable stages (\FIG{overview}):
\begin{itemize}
    \item \textbf{Stage 1:} Continuous atom embeddings are constructed using graph neural networks to perceive chemical environments (Section~\ref{sec:stage1}).
    \item \textbf{Stage 2:} Continuous bond, angle, and torsion embeddings are constructed using pooling that preserves appropriate symmetries (Section~\ref{sec:stage2}).
    \item \textbf{Stage 3:} Molecular mechanics force field parameters are computed from atom, bond, angle, and torsion embeddings using feed-forward networks (Section~\ref{sec:stage3}).
\end{itemize}
Additional molecular mechanics parameter classes (such as point polarizabilities, valence coupling terms, or even parameters for charge-transfer models~\cite{ko2021fourth}) can easily be added in a modular manner.

Similar to legacy molecular mechanics parameter assignment infrastructures, molecular mechanics parameters are assigned \textit{once} for each system, and can be subsequently used to compute energies and forces or carry out molecular simulations with standard molecular mechanics packages.
Unlike traditional legacy force fields, espaloma model parameters $\Phi_\mathrm{NN}$---which define the entire \emph{process} by which molecular mechanics force field parameters $\Phi_\mathrm{FF}$ are generated \textit{ad hoc} for a given molecule---can easily be fit to data from scratch using standard, highly portable, high-performance machine learning frameworks that support automatic differentiation.

Here, we demonstrate that espaloma provides a sufficiently flexible representation to both learn to apply existing MM force fields and to generalize them to new molecules (Section~\ref{sec:mm_fitting}).
Espaloma's modular loss function enables force fields to be learned directly from quantum chemical energies (Section~\ref{sec:qm-fitting}), partial charges (Section~\ref{sec:charge-equilibration}), or both.
The resulting force fields can generate self-consistent parameters for small molecules, biopolymers~(Section~\ref{sec:biopolymers}), and covalent adducts~(Section~\ref{sec:covalent-ligands}).
Finally, an espaloma model fit to the same quantum chemical dataset used to build the Open Force Field OpenFF ("Parsley") \textsc{openff-1.2.0} small molecule force field, augmented with peptide quantum chemical data, can outperform it in an out-of-sample kinase:inhibitor alchemical free energy benchmark (Section~\ref{sec:md-simulation-details}).


\section{Espaloma: End-to-end differentiable MM parameter assignment}
\label{sec:end-to-end-differentiable-approach}
First, we describe how our proposed framework, {\bf espaloma} (\FIG{overview}), operates analogously to legacy force field typing schemes to generate MM parameters $\Phi_\text{FF}$ from a molecular graph $\mathcal{G}$ and neural model parameters $\Phi_\text{NN}$,
\begin{equation}
\Phi_\text{FF} \leftarrow \texttt{espaloma} (\mathcal{G}, \Phi_\text{NN}),
\end{equation}
which can subsequently be used to compute the MM energy (as in Equation~\ref{mm_graph}) for any conformation.
A brief graph-theoretic overview of molecular mechanics force fields is provided in the Appendix (Section~\ref{sec:molecular-mechanics-forcefields}).


\subsection{\textbf{Stage 1}: Graph neural networks generate a continuous atom embedding, replacing legacy discrete atom typing}
\label{sec:stage1}
We propose to use graph neural networks~\cite{xu2018powerful, wang2019graph, DBLP:journals/corr/KipfW16, battaglia2018relational, Du:ArXiv171010370CsStat:2018, wu2019simplifying, wang2019deep, wang2019dynamic, duvenaud2015convolutional, gilmer2017neural} as a replacement for rule-based chemical environment perception~\cite{mobley2018escaping}.
These neural architectures learn useful representations of atomic chemical environments from simple input features by updating and pooling embedding vectors via message passing iterations to neighboring atoms ~\cite{gilmer2017neural}.
As such, symmetries in chemical graphs (chemical equivalencies) are inherently preserved, while a rich, continuous, and differentiably learnable representation of the atomic environment is derived.
For a brief introduction to graph neural networks, see Appendix Section~\ref{sec:introduction-to-graph-nets}

Traditional molecular mechanics force field parameter assignment schemes such as Antechamber/GAFF~\cite{wang2004development,wang2006automatic} or CGenFF~\cite{vanommeslaeghe2012automation,vanommeslaeghe2012automation2} use attributes of atoms and their neighbors (such as atomic number, hybridization, and aromaticity) to assign discrete atom types. Subsequently, atom, bond, angle, and torsion parameters are assigned for specific combinations of these discrete types according to human chemical intuition~\cite{mobley2018escaping}.
On a closer look, this scheme resembles a two- or three-step Weisfeiler-Leman test~\cite{weisfeiler1968reduction}, which has been shown to be well approximated by some graph neural network architectures~\cite{xu2018powerful}.
We hypothesize that graph neural network architectures can be at least as expressive as legacy atom typing rules.

To compute continuous atom embeddings, we start with a molecular graph $\mathcal{G}$ whose atoms (nodes) are labeled with simple chemical properties (here, we consider element, hybridization, aromaticity, and formal charge, and membership in various ring sizes) easily computed in any cheminformatics toolkit.
Sequential application of the graph neural network message-passing update rules (Section~\ref{sec:introduction-to-graph-nets}) then computes an updated set of atom (node) features in each graph neural network layer, and the final atom embeddings $h_v \in \mathbb{R}^{\mid \mathcal{G} \mid \times D}$ are extracted from the final layer.
The loss on training data is then minimized by minimizing the cross-entropy loss between predicted and reference types.\footnote{
Note that this discrete type assignment layer is only used to address the question of how well the continuous embeddings approximate discrete types, and is not used in subsequent experiments that utilize the standard espaloma architecture (\FIG{overview}).}

We use a subset of ZINC~\cite{irwin2005zinc} provided with parm@Frosst to validate atom typing implementations~\cite{parm_frosst} (7529 small drug-like molecules, partitioned 80:10:10 into training:validation:test sets) for this experiment.
Reference GAFF 1.81~\cite{wang2006automatic} atom types are assigned using Antechamber~\cite{wang2006automatic} from AmberTools and are used for training and testing.
\begin{figure}[h]
    \centering
    \includegraphics[width=0.9\textwidth]{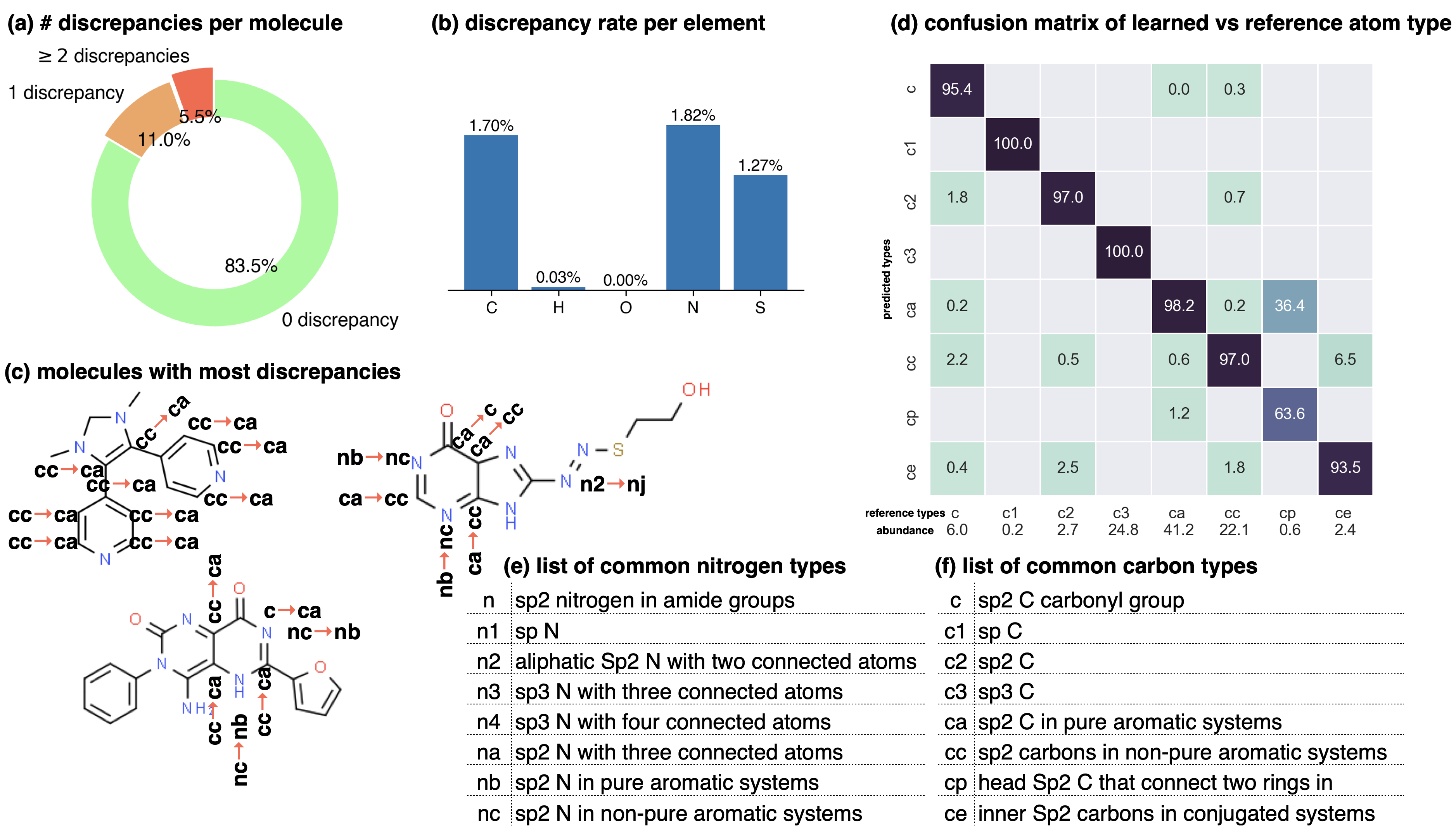}
    \caption{\textbf{Graph neural networks can reproduce legacy atom types with high accuracy.}
    The Stage 1 graph neural network of espaloma (Section~\ref{sec:stage1}) chained to a discrete atom type readout was fit to GAFF~1.81 atom types~\cite{wang2004development,wang2006automatic} on a subset of ZINC~\cite{irwin2005zinc} distributed with parm@Frosst~\cite{parm_frosst} as a validation set.
    The 7529 molecules in this set were partitioned 80:10:10 into training:validation:test sets for this experiment.
    Within test set, $99.07\%_{98.93\%}^{99.22\%}$ atoms were correctly typed, with 1000 bootstrap replicates used to estimate the confidence intervals arising from finite test set size effects.
    \emph{(a)} The distribution of the number of atom type discrepancies per molecule on the test set demonstrates that only a minority of atoms are incorrectly typed.
    \emph{(b)} The error rate per element is primarily concentrated within carbon, nitrogen, and sulfur types.
    \emph{(c)} Examining atom type failures in detail on molecules with the largest numbers of discrepancies shows that the atom types are easily confusing even to a human expert, since they represent qualities that are difficult to precisely define.
    \emph{(d)} The distribution of predicted atom types for each reference atom type for carbon types are shown; on-diagonal values indicate agreement. 
    The percentages annotated under x-axis denote the relative abundance within the test set.
    Only the common carbon types are included in the confusion matrix here; for full confusion matrix across all atom types, see SI Figure~\ref{extra-confusion}.
    (e) A list of common nitrogen types in GAFF-1.81~\cite{doi:10.1002/jcc.20035}.
    (f) A list of common carbon types in GAFF-1.81~\cite{doi:10.1002/jcc.20035}. 
    }
    \label{fig:atom-typing-accuracy}
\end{figure}

\paragraph{Graph neural networks can reproduce legacy atom types with high accuracy}

The test set performance is reported in Figure~\ref{fig:atom-typing-accuracy}, where the overall accuracy between reference legacy types and learned types is very high---$98.31\%_{97.94\%}^{98.63\%}$. 
In analyzing the infrequent failures, we find the model assigns atom types that correspond to the reference type more often when the atom type appears more frequently in the training data, whereas the discrepancies occur in assigning rare types and types whose definitions follow a more sophisticated (but chemically arbitrary) logic.
For instance, one of the most frequent confusions is the misassignment of {\tt ca} (sp2 carbon in pure aromatic systems) to {\tt cp} (head sp2 carbon that connects two rings in biphenyl systems, occurring in only 0.6\% of the dataset).
The relative ambiguity of the various types that are most frequently confused is suggestive that the graph net makes human-like errors in perceiving subtle differences between distinct chemical environments.

The benefits of neural embedding compared to legacy discrete typing are many-fold:
\begin{itemize}
    \item Legacy typing schemes are generally described in text form in published work (for example~\cite{wang2004development,wang2006automatic}), creating the potential for discrepancies between implementations when different cheminformatics toolkits are used.
    By contrast, with the \textit{knowledge} to distinguish chemical environments stored in latent vectors and not dependent on any manual coding, our approach is deterministic once trained, and is portable across platforms thanks to modern machine learning frameworks.
    \item Both the chemical perception process and the application of force field parameters $\Phi_\mathtt{FF}$ can be optimized simultaneously via gradient-based optimization of $\Phi_\text{NN}$ using standard machine learning frameworks that support automatic differentiation. 
    \item While extending a legacy force field by adding new atom types can lead to an explosion in the number of parameter types, 
    continuous neural embeddings do not suffer from this limitation; expansion of the typing process occurs automatically as more diverse training examples are introduced.
\end{itemize}

\subsection{\textbf{Stage 2}: Symmetry-preserving pooling generates continuous bond, angle, and torsion embeddings, replacing discrete types}
\label{sec:stage2}

Terms in a molecular mechanics potential are symmetric with respect to certain permutations of the atoms involved in the interaction.
For example, harmonic bond terms are symmetric with respect to the exchange of atoms involved in the bond.
More elaborate symmetries are frequently present, such as in the three-fold terms representing improper torsions for the Open Force Field "Parsley" generation of force fields ($\mathtt{k-i-j-l}$, $\mathtt{k-j-l-i}$, and $\mathtt{k-l-i-j}$, where $k$ is the central atom)~\cite{qiu_smith2021}.
Traditional force fields, for bond, angle, and proper torsion terms, enforce this by ensuring equivalent orderings of atom types receive the same parameter value.
\footnote{Traditional force fields group bonds, angles, and torsions simply by their composing ordered groups of atoms.
For instance, the first bond type in GAFF 1.81~\cite{wang2006automatic} is defined by the types {\tt hw-ow}\cite{gaff-dat}, which is equivalent to {\tt ow-hw} due to mirror symmetry in identifying bonds.
Angles and torsions have similar symmetries that must to be accounted for when enumerating the atoms or matching valence types.
Note that Amber does not uniquely specify equivariant improper torsion orderings---see footnote \emph{a} of Table 3 of \citep{mobley2018escaping} for details.
}

For neural embeddings, the invariances of valence terms w.r.t. these atom ordering symmetries must be considered while searching for expressive latent representations to feed into a subsequent parameter prediction network stage.
Inspired by Janossy pooling~\cite{DBLP:journals/corr/abs-1811-01900}, we enumerate the the relevant equivalent atom permutations to derive bond, angle, and torsion embeddings $h_r, h_\theta, h_\phi$ that respect these symmetries from atom embeddings $h_v$, 
\begin{gather}
    h_{r_{ij}} =  \operatorname{NN}_r ([h_{v_i}:h_{v_j}]) + \operatorname{NN}_r ([h_{v_j}:h_{v_i}]);\\
    h_{\theta_{ijk}} = \operatorname{NN}_\theta ([h_{v_i} : h_{v_j} : h_{v_k}]) + \operatorname{NN}_\theta ([h_{v_k} : h_{v_j} : h_{v_i}]);\\
    h_{\phi_{ijkl_\text{proper}}} = \operatorname{NN}_{\phi_\text{proper}} ([h_{v_i} : h_{v_j} : h_{v_k} : h_{v_l}]) + \operatorname{NN}_\phi ([h_{v_l} : h_{v_k} : h_{v_j} : h_{v_i}]);\\
    h_{\phi_{ijkl_\text{improper}}} = \operatorname{NN}_{\phi_\text{improper}} ([h_{v_i} : h_{v_j} : h_{v_k} : h_{v_l}]) + \operatorname{NN}_\phi ([h_{v_k} : h_{v_j} : h_{v_l} : h_{v_i}]) + \operatorname{NN}_{\phi_\text{improper}} ([h_{v_l} : h_{v_j} : h_{v_i} : h_{v_k}]),
\end{gather}
where columns $(\cdot : \cdot)$ denote concatenation~\footnote{Here, we use the threefold improper formulation used by the Open Force Field "Parsley" generation force fields, which avoids the ambiguity associated with selecting a single arbitrary improper torsion from a set of four atoms involved in the torsion~\cite{qiu_smith2021}.
}. 
As such, the order invariance is evident, i.e., $h_{r_{ij}}=h_{r_{ji}}, h_{\theta_{ijk}}=h_{\theta_{kji}},$ and $h_{\phi_{ijkl}}=h_{\phi_{lkji}}$\footnote{In SI Section~\ref{sec:janossy-works}, we prove that this form is sufficiently expressive to assign unique valence types.}.

\subsection{\textbf{Stage 3}: Neural parametrization of atoms, bonds, angles, and torsions replaces tabulated parameter assignment}
\label{sec:stage3}
In the final stage, each feed-forward neural networks modularly learn the mapping from these symmetry-preserving atom, bond, angle, and torsion encodings to MM parameters $\Phi_\text{FF}$ that reflect the specific chemical environments appropriate for these terms:
\begin{eqnarray}
 \{ \epsilon_v, \sigma_v\} = \operatorname{NN}_{v_\text{readout}}(h_v) & \:\: \textbf{atom parameters} \\
\{ k_r, b_r \} = \operatorname{NN}_{r_\text{readout}}(h_r) & \:\: \textbf{bond parameters} \\ 
\{ k_\theta, b_\theta \} = \operatorname{NN}_{\theta_\text{readout}}(h_\theta) & \:\: \textbf{angle parameters} \\ 
\{ k_\phi \} = \operatorname{NN}_{\phi_\text{readout}}(h_\phi) & \:\: \textbf{torsion parameters}
\end{eqnarray}
This stage is analogous to the final table lookup step in traditional force field construction, but with significant added flexibility arising from the continuous embedding that captures the chemical environment specific to the potential energy term being assigned.

Here, we use Lennard-Jones parameters from legacy force fields (here, Open Force Field 1.2.0~\cite{openff-1.2.0}) to avoid having to include condensed-phase physical properties in the fitting procedure.
While including condensed-phase physical properties in the loss function is possible, it is very expensive to do so, and as our experiments demonstrate, may not be necessary for achieving increased accuracy over legacy force fields.

We also found producing bond and angle parameters directly in Stage 3 to frustrate optimization, so we employ a mixture of linear bases to represent harmonic energies that can be translated back to the original functional form (see Appendix Section~\ref{sec: linear_basis}).
Similarly, we do not fit phases and periodicities of torsions as they are discrete. 
We instead fix phases at $\phi_0=0$ and fit all periodicities $n = 1, \ldots, 6$.
This allows the corresponding torsion barriers $K_n$ to assume the entire continuum of positive or negative values; as a result, $K_n<0$ mimics the effect of $\phi_0 = \pi$.

As a result of using the continuous atom embedding vectors to represent chemical environments for each atom, it is possible to intelligently interpolate between relevant chemical environments seen during training. 
This interpolation produces more nuanced varieties of parameters than either traditional atom typing or direct chemical perception, and is capable of capturing subtle effects arising from fractional bond order perturbation~\cite{stern2020capturing}.
Due to the modularity of this stage, it is easy to add new modules or swap out existing ones to explore other force field functional forms, such as alternative vdW interactions~\cite{halgren1992representation}; pair-specific Lennard-Jones interaction parameters~\cite{baker2010accurate,doi:10.1021/acs.jctc.2c00115}; point polarizabilities for instantaneous dipole~\cite{ren2002consistent}, Drude oscillator~\cite{lemkul2016empirical}, or Gaussian charge~\cite{leven2019c} polarizability models; class II valence couplings~\cite{maple1994derivation-1,hwang1994derivation-2,maple1994derivation-3}; charge transfer~\cite{ko2021fourth}; or other potential energy terms of interest.


\begin{figure}
    \centering
    \includegraphics[width=1.0\textwidth]{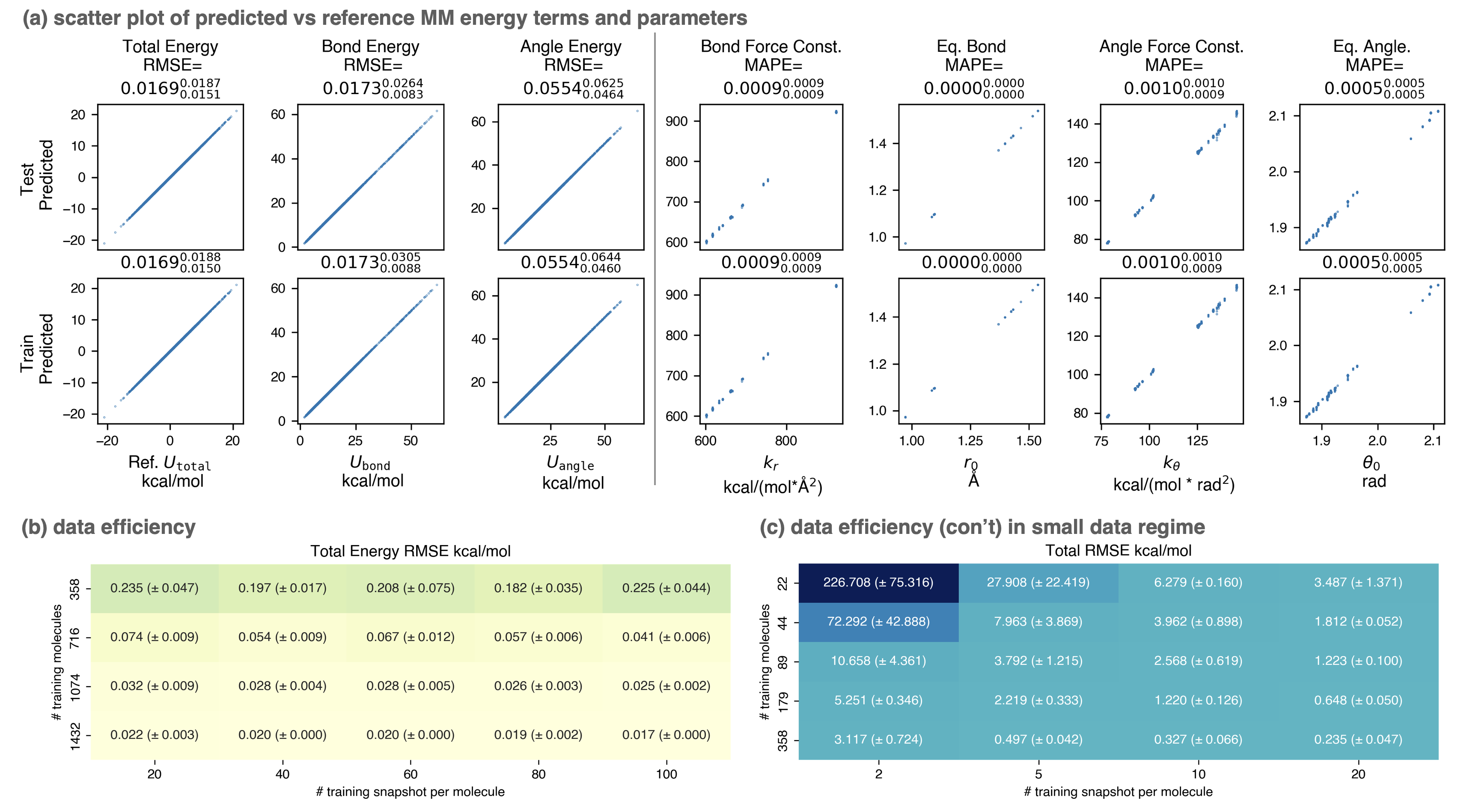}
    \caption{\textbf{Espaloma accurately learns molecular mechanics parameters when fit to snapshot energies from a molecular mechanics force field.}
    In this experiment, espaloma was used to fit molecular mechanics (GAFF~1.81) potential energies of snapshots generated from short molecular dynamics (MD) simulations initiated from multiple conformers of molecules from the PhAlkEthOH dataset, which uses only the elements carbon, oxygen, and hydrogen~\cite{PhAlkEthOH-1.0}.
    The dataset contains 7408 molecules with 100 snapshots each, and was partitioned by molecules 80:10:10 into train:validate:test sets.
    We excluded three- and four-membered rings in the dataset; for a detailed study on the roles of these two factors, see SI Section~\ref{sec:close-examination}.
    (a) Statistics quoted above the plots provide the root mean squared error (RMSE) between reference and predicted MM energies, and mean absolute percentage error (MAPE) (in fractional form) between reference and predicted force field parameters. 
    The sub- and superscripts report the the 95\% confidence interval of each statistics estimated from 1000 bootstrapped replicates over molecules in the test set.
    Energy terms have kcal/mol units whereas the units of force field parameters does not affect statistics reported here.
    Force field parameters ($\Phi_\text{FF}$):
    $K_r$: bond force constant;
    $b_r$: equilibrium bond length;
    $K_\theta$: angle force constant;
    $b_\theta$: equilibrium angle value.
    Torsion parameters are not shown because of the potential for degeneracy of fit given that all periodicities $n=1, 2, \ldots, 6$ are learned by espaloma.
    (b) The data efficiency of espaloma was assessed in a typical use-case regime by exploring the test set energy RMSE as a function of the number of training molecules and snapshots per molecule.
    The standard deviation over three fitting experiments with different random seeds is shown in parenthesis.
    Once a sufficiently large number of molecules are available, doubling the number of snapshots per molecule does not reduce the error as rapidly as doubling the number of molecules.
    (c) The data efficiency of espaloma in a data-poor regime was assessed in the same manner as (a), but for a small number of molecules and training snapshots per molecule.
    In the data-poor regime, increasing both the number of molecules and snapshots per molecule can deliver large decreases in test set error.
    }
    \label{fig:mm_fitting}
\end{figure}

\section{Espaloma can learn to mimic existing molecular mechanics force fields from snapshots and associated potential energies}
\label{sec:mm_fitting}

Having established that graph neural networks have the capacity to learn to reproduce legacy atom types describing distinct chemical environments, we ask whether espaloma is capable of learning to reproduce traditional molecular mechanics (MM) force fields assigned via standard atom typing schemes.
In addition to quantifying how well a force field can be learned when the exact parameters of the model being learned are known, being able to accurately learn existing MM force fields would have numerous applications, including replacing legacy non-portable parameter assignment codes with modern portable machine learning frameworks, learning to generalize to new molecules that contain familiar chemical environments, and permitting simplified parameter assignment for complex, heterogeneous systems involving post-translational modifications, covalent ligands, or heterogeneous combinations of biopolymers and small molecules.

To assess how well espaloma can learn to reproduce a molecular mechanics force field from a limited amount of data, we selected a dataset with limited chemical complexity---PhAlkEthOH~\cite{Bannan2019,PhAlkEthOH-1.0}---which consists of 7408 linear and cyclic molecules containing phenyl rings, small alkanes, ethers, and alcohols composed of only the elements carbon, oxygen, and hydrogen.
Three- and four-membered rings are excluded in the dataset since they would cause instability in the prediction of energies (see SI Section~\ref{sec:close-examination}).
We generated a set of 100 conformational snapshots for each molecule using short molecular dynamics simulations at 300~K initiated from multiple conformations to ensure adequate sampling of conformers.
The PhAlkEthOH dataset was randomly partitioned (by molecules) into 80\% training, 10\% validation, and 10\% test molecules, and an espaloma model was trained with early stopping via monitoring for a decrease in accuracy in the validation set.
The performance of the resulting model is shown in \FIG{mm_fitting}.

\paragraph{Espaloma can learn existing force fields and generalize to new molecules with low error}

Espaloma is able to achieve very low total energy and parameter error on the training set, suggesting that espaloma can learn the parameters of typed molecules from energies alone.
In addition, error on the out-of-sample test set of molecules is comparable---less than 0.02 kcal/mol---suggesting that espaloma can effectively generalize to new molecules within the same chemical space. 
Surprisingly, the total energy RMSE is lower than the angle energy RMSE, suggesting that there is some degeneracy in how energy contributions are distributed among valence energy terms.

\paragraph{Espaloma requires few conformations per molecule to achieve high accuracy}

We examined the data efficiency of espaloma by repeating the MM fitting experiment with varying numbers of molecules and snapshots per molecule in an attempt to address whether \emph{molecular diversity} or \emph{conformational diversity} is more important.
In the typical data regime (\FIG{mm_fitting}b), once sufficient conformational diversity is reached ($\sim$20 snapshots/molecule), increasing molecular diversity more effectively reduces error, though this meets with diminishing returns past a certain point.
In the low data regime (\FIG{mm_fitting}c), both molecular and conformational diversity are important for reducing error to useful regimes, with a minimal threshold for each required to achieve reasonable errors.

\section{Espaloma can fit quantum chemical energies directly to build new molecular mechanics force fields}
\label{sec:qm-fitting}


\begin{figure}[tbp]
    \centering
    \resizebox{\textwidth}{!}{%
    \begin{tabular}{c c c c c c c c c c c}
    \hline
    
    & \multirow{2}{*}{{\bf (a)} dataset}
    & \multirow{2}{*}{\# mols}
    & \multirow{2}{*}{\# trajs}
    & \multirow{2}{*}{\# snapshots}
    & \multicolumn{2}{c}{Espaloma RMSE}
    & \multicolumn{4}{c}{Legacy FF RMSE (kcal/mol) (Test molecules)}\\
    
    & & & &
    & Train
    & Test
    & OpenFF~1.2.0
    & GAFF-1.81 
    & GAFF-2.11 
    & Amber~ff14SB\\
    \hline
    
    \multicolumn{2}{c}{\textbf{PhAlkEthOH} (simple CHO)}
    & 7408 & 12592 & 244036
    & $0.8656_{0.8225}^{0.9131}$
    & $1.1398_{1.0715}^{1.2332}$
    & $1.6071_{1.5197}^{1.6915}$
    & $1.7267_{1.6543}^{1.7935}$
    & $1.7406_{1.6679}^{1.8148}$\\
    
    \multicolumn{2}{c}{\textbf{OpenFF Gen2 Optimization} (druglike)}
    & 792 & 3977 & 23748
    & $0.7413_{0.6914}^{0.7920}$
    & $0.7600_{0.6644}^{0.8805}$
    & $2.1768_{2.0380}^{2.3388}$
    & $2.4274_{2.3300}^{2.5207}$
    & $2.5386_{2.4370}^{2.6640}$ \\
    
    \multicolumn{2}{c}{\textbf{VEHICLe} (heterocyclic)}
    & 24867 & 24867 & 234326
    & $0.4476_{0.4273}^{0.4690}$
    & $0.4233_{0.4053}^{0.4414}$
    & $8.0247_{7.8271}^{8.2456}$
    & $8.0077_{7.7647}^{8.2313}$
    & $9.4014_{9.2135}^{9.6434}$ \\
    
    \multicolumn{2}{c}{\textbf{PepConf} (peptides)}
    & 736 & 7560 & 22154
    & $1.2714_{1.1899}^{1.3616}$
    & $1.8727_{1.7309}^{1.9749}$
    & $3.6143_{3.4870}^{3.7288}$
    & $4.4446_{4.3386}^{4.5738}$
    & $4.3356_{4.1965}^{4.4641}$ 
    & $3.1502_{3.1117}^{3.1859, *}$\\
    \hline
    
    \multirow{2}{*}{\textbf{joint}}
    & OpenFF Gen2 Optimization
    & \multirow{2}{*}{1528}
    & \multirow{2}{*}{11537}
    & \multirow{2}{*}{45902}
    & $0.8264_{0.7682}^{0.9007}$
    & $1.8764_{1.7827}^{1.9947}$
    & $2.1768_{2.0380}^{2.3388}$
    & $2.4274_{2.3300}^{2.5207}$
    & $2.5386_{2.4370}^{2.6640}$ \\
    
    & PepConf
    & & &
    & $1.2038_{1.1178}^{1.3056}$
    & $1.7307_{1.6053}^{1.8439}$
    & $3.6143_{3.4870}^{3.7288}$
    & $4.4446_{4.3386}^{4.5738}$
    & $4.3356_{4.1965}^{4.4641}$ 
    & $3.1502_{3.1117}^{3.1859, *}$\\
    \hline
    
    \end{tabular}}
    \includegraphics[width=1.0\textwidth]{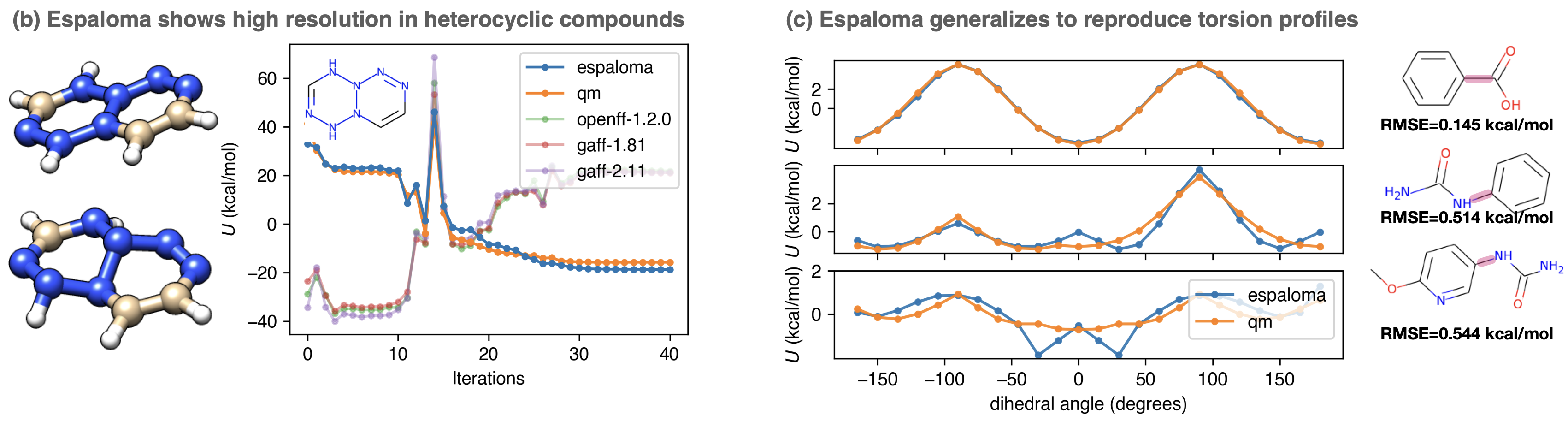}
    \caption{\textbf{Espaloma can directly fit quantum chemical energies to produce new molecular mechanics force fields with better accuracy than traditional force fields based on atom typing or direct chemical perception.}
    Espaloma was fit to quantum chemical potential energies for conformations generated by optimization trajectories initiated from distinct conformers in various datasets from QCArchive~\cite{wang2006automatic}.
    All datasets were partitioned by molecules 80:10:10 into train:validate:test sets.
    The number of molecules, optimization trajectories, and total snapshots are annotated in the table.
    (a) We report the RMSE on training and test sets, as well as the performance of legacy force fields on the test set.
    All statistics are computed with predicted and reference energies centered to have zero mean for each molecule, in order to focus on errors in relative conformational energetics rather than on errors in predicting the heats of formation of chemical species (which the MM functional form used here is incapable of).
    The 95\% confidence intervals annotated are calculated by via bootstrapping molecules with replacement using 1000 replicates.
    (b) Optimization trajectory of a representative (wih highest OpenFF 1.2.0 RMSE) heterocyclic compound in VEHICLe dataset with SMILES string \textsc{[H]C1=C(N2N(N=C(N(N2N=N1)[H])[H])[H])[H]}. 
    Legacy force fields, because of their limited chemical typing rules, were not able to perceive the chemical environment of the nitrogen atoms, which were not aromatic.
    (c) Espaloma is able to predict energies for quantum chemical torsion scans for an out-of-sample torsion scan dataset (the OpenFF Phenyl Torsion Drive Dataset~\cite{stern2020capturing, phenyl}, dihedral angle profiled marked in rouge.) to high accuracy even though it was not trained on torsion scans (only optimization trajectories) or any of the molecules in the torsion scan set.
    *: Six cyclic peptides that cannot be parametrized using OpenForceField toolkit engine~\cite{openff-toolkit-0.10.0} are not included.
    }
    \label{fig:qm-fitting}
\end{figure}


Since espaloma can derive a force field solely by fitting to energies (and optionally gradients), we repeat the end-to-end fitting experiment (Section~\ref{sec:mm_fitting}) directly using quantum chemical (QM) datasets used to build and evaluate MM force fields.
We assessed the ability of Espaloma to learn several distinct quantum chemical datasets generated by the Open Force Field Initiative~\cite{mobley2018open} and deposited in the MolSSI QCArchive~\cite{smith2020molssi} with B3LYP-D3BJ/DZVP level of theory:
\begin{itemize}
    \item \textbf{PhAlkEthOH}~\cite{PhAlkEthOH-1.0} is a collection of compounds containing only the elements carbon, hydrogen, and oxygen in compounds containing phenyl rings, alkanes, ketones, and alcohols.
    Limited in elemental and chemical diversity, this dataset is chosen as a proof-of-concept to demonstrate the capability of espaloma to fit and generalize quantum chemical energies when training data is sufficient to exhaustively cover the breadth of chemical environments.
    \item \textbf{OpenFF Gen2 Optimization}~\cite{gen2} consists of druglike molecules used in the parametrization of the Open Force Field 1.2.0 ("Parsley") small molecule force field~\cite{qiu2021development}. 
    This set was constructed by the Open Force Field Consortium from challenging molecule structures provided by Pfizer, Bayer, and Roche, along with diverse molecules selected from eMolecules to achieve useful coverage of chemical space. 
    \item \textbf{VEHICLe}~\cite{vehicle} or \textit{virtual exploratory heterocyclic library}, is a set of heteroaromatic ring systems of interest to drug discovery enumerated by \citet{doi:10.1021/jm801513z}. 
    The atoms in the molecules in this dataset have interesting chemical environments in heteroarmatic rings that present a challenge to traditional atom typing schemes, which cannot easily accommodate the nuanced distinctions in chemical environments that lead to perturbations in heterocycle structure.
    We use this dataset to illustrate that espaloma performs well in situations challenging to traditional force fields.
    \item \textbf{PepConf}~\cite{pepconf} from \citet{prasad_otero-de-la-roza_dilabio_2019} contains a variety of short peptides, including capped, cyclic, and disulfide-bonded peptides.
    This dataset---regenerated as an OptimizationDataset (quantum chemical optimization trajectories initiated from multiple conformers) using the Open Force Field QCSubmit tool~\cite{qcsubmit}---explores the applicability of espaloma to biopolymers, such as proteins.
\end{itemize}


Since nonbonded terms are generally optimized to fit other condensed-phase properties, we focused here on optimizing only the valence parameters (bond, angle, and proper and improper torsion) to fit these gas-phase quantum chemical datasets, fixing the non-bonded energies using a legacy force field~\cite{mobley2018open}.
In this experiment, all the non-bonded energies (Lennard-Jones and electrostatics) were computed using Open Force Field 1.2 Parsley~\cite{jeff_wagner_2020_4021623}, with AM1-BCC charges generated by the OpenEye Toolkit back-end for the Open Force Field toolkit 0.10.0~\cite{openff-toolkit-0.10.0}.
Because we are learning an MM force field that is incapable of reproducing quantum chemical heats of formation, which are reflected as an additive offset in the quantum chemical energy targets, snapshot energies for each molecule in both the training and test sets are shifted to have zero mean.
All datasets are randomly shuffled and split (by molecules) into training (80\%), validation (10\%), and test (10\%) sets.

\paragraph{Espaloma generalizes to new molecules better than widely-used traditional force fields}

To assess how well espaloma is able to generalize to new molecules, the performance for espaloma on test (and training) sets was compared to a legacy atom typing based force field (GAFF 1.81 and 2.11~\cite{wang2004development,wang2006automatic}, which collectively have been cited over 13,066 times) and a modern force field based on direct chemical perception~\cite{mobley2018escaping} (the Open Force Field 1.2.0 ("Parsley") small molecule force field~\cite{openff-1.2.0}, downloaded over 150,000 times).

The results of this experiment are reported in \FIG{qm-fitting}.
As can be readily seen by the reported test set root mean squared error (RMSE), espaloma can produce MM force fields with generalization performance consistently better than legacy force fields based on discrete atom typing (GAFF~\cite{wang2004development,wang2006automatic}).
In chemically well-represented datasets like PhAlkEthOH---which contains only simple molecules constructed from elements C, H, and O---espaloma is able to significantly improve on the accuracy of traditional force fields such as OpenFF~1.2.0, GAFF-1.81, and GAFF-2.11 on the test set.

Surprisingly, even though OpenFF 1.2.0 included the "Open FF Gen 2" dataset in training, espaloma is able to achieve superior \emph{test} set performance on this dataset, suggesting that both the flexibility and generalizability of continuous atom typing have significant advantages over even direct chemical perception~\cite{mobley2018escaping}.

Even compared to highly optimized late-generation \emph{protein} force fields such as Amber ff14SB~\cite{amber2020}---which was highly optimized to reproduce quantum chemical torsion drive data---espaloma achieves significantly higher accuracy, improving on Amber ff14SB error of $3.1502_{3.1117}^{3.1859, *}$ kcal/mol to achieve $1.8727_{1.7309}^{1.9749}$ kcal/mol on the PepConf peptide dataset~\cite{pepconf, prasad2019pepconf}.
This suggests that espaloma is capable of effectively parameterizing both small molecule and biopolymer force fields.
Indeed, when we train an espaloma model using \emph{both} the OpenFF Gen2 Optimization and PepConf datasets ({\bf Joint} in \FIG{qm-fitting}(a)), we see that a single espaloma model is capable of achieving superior accuracy to traditional small molecule \emph{and} protein force fields simultaneously.

\paragraph{Espaloma can automatically learn distinct atom environments overlooked by traditional force fields}

It is worth noting that that the traditional but widely used force fields considered here uniformly perform poorly on the VEHICLe dataset~\cite{doi:10.1021/jm801513z} ("Heteroaromatic Rings of the Future", containing heterocyclic scaffolds of interest to future drug discovery programs).
In \FIG{qm-fitting}(b), we show the most common mode of failure of legacy force fields by examining their predicted energy over the QM optimization trajectory of the compound with largest RMSE (with SMILES string \textsc{[H]C1=C(N2N(N=C(N(N2N=N1)[H])[H])[H])[H]}).
The initial conformation of the molecule, generated by OpenEye Toolkit, was planar.
As the conformation was optimized by quantum chemical methods, the tertiary nitrogens in the system become pyramidal.
In a closer examination, GAFF-2.11, for instance, assigned all carbons to be of type \textsc{cc} and all nitrogens \textsc{na}, indicating that they were perceived as aromatic, whereas there is no conjugated system present in the molecule.
This also reflects the limitation in \textit{resolution} of legacy force fields. 
Espaloma, on the other hand, provides a high-resolution atom embedding that can flexibly characterize the chemical environments, provided that similar environments existed in the training data.

\paragraph{Espaloma can reliably learn torsion profiles from optimization trajectories}

We wondered whether espaloma could faithfully recover torsion energy profiles---which are traditionally expensive to generate using methods like wavefront propagation~\cite{qiu2020driving}---from the inexpensive optimization trajectories used to train espaloma models. 
We therefore examined some representative dihedral energy profiles for molecules outside of the dataset used to train espaloma.
In \FIG{qm-fitting}(c), we use the espaloma model trained on OpenFF Gen2 Optimization and PepConf to predict the energy profiles of several torsion drive experiments in the OpenFF Phenyl Torsion Drive Dataset~\cite{phenyl}---which does not contain any of the molecules in the training set---and observed that the locations and heights of torsion energy barriers are recapitulated with reasonable accuracy.
This suggests that optimization trajectories are sufficient to capture the locations and relative heights of torsion barriers---a highly useful finding given the relative expense of generating accurate torsion profiles compared to simple optimization trajectories~\cite{qiu2020driving}.

\section{Espaloma can learn self-consistent charge models in an end-to-end differentiable manner}
\label{sec:charge-equilibration}


Historically, biopolymer force fields derive partial atomic charges via fits to high-level multiconformer quantum chemical electrostatic potentials on capped model compounds, adjusted to ensure the repeating biopolymer units have integral charge (often incorporating constraints to share identical backbone partial charges)~\cite{cornell1995second,duan2003point,best2012optimization}.
Some approaches to the derivation of partial atomic charges are enormously expensive, requiring iterative QM/MM simulations in explicit solvent to derive partial charges for new molecules~\cite{cerutti2013derivation,debiec2016further}.
For small molecules, state-of-the-art methods range from fast bond charge corrections applied to charges derived from semiempirical quantum chemical methods (such as AM1-BCC~\cite{jakalian2000fast,jakalian2002fast} or CGenFF charge increments~\cite{vanommeslaeghe2012automation}) to expensive multiconformer restrained electrostatic potential (RESP) fits to high-level quantum chemistry~\cite{bayly1993well,schauperl2020non}.
Surprisingly little attention has been paid to the divergence of methods used for assigning partial charges to small molecules and biopolymers, and the potential impact this inconsistency has on accuracy or ease of use---indeed, developing charges for post-translational modifications to biopolymer residues~\cite{khoury2013forcefield_ptm, atz_isert_bocker_jimenez-luna_schneider_2021} or covalent ligands can prove to be a significant technical challenge in attempting to bridge these two worlds.

While machine learning approaches have begun to find application in determining small molecule partial charges~\cite{bleiziffer2018machine,lubbers2018hierarchical,sifain2018discovering}, methods such as random forests are not fully continuously differentiable, rendering them unsuitable for a  fully end-to-end differentiable parameter assignment framework.
Recently, a fast (500x speed up for small molecules) approach has been proposed that uses graph neural networks as part of a charge-equilibration~\cite{rappe1991charge,doi:10.1021/ja00275a013} scheme (inspired by the earlier VCharge model~\cite{gilson2003fast}) to self-consistently assign partial charges to small molecules, biopolymers, and arbitrarily complex hybrid molecules in a conformation-independent manner that only makes use of molecular topology~\cite{wang2019graph}.
Perhaps unsurprisingly, due to the requirement that molecules retain their integral net charge, directly predicting partial atomic charges from latent atom embeddings and subsequently renormalizing charges leads to poor performance (0.28 e~\cite{wang2019graph}).

\begin{table}[tb]
    \centering
        
        
        
    \resizebox{\textwidth}{!}{%
    \begin{tabular}{c c | c c c c }
    \hline
        \multicolumn{2}{c}{\bf PhAlkEthOH experiment (combination of independent models)} &
        \multicolumn{2}{c}{\bf energy RMSE (kcal/mol)} &
        \multicolumn{2}{c}{\bf charge RMSE (e)} \\
        valence force field &
        charge model & 
        Train & Test & Train & Test \\
        \hline 
        openff-1.2.0 & AM1-BCC &  \multicolumn{2}{c}{$1.6071_{1.5197}^{1.6915}$} & \multicolumn{2}{c}{reference} \\
        openff-1.2.0 & espaloma & $1.6286_{1.5861}^{1.6628}$ & $1.7072_{1.6361}^{1.7913}$ & $0.0072_{0.0071}^{0.0073}$ & $0.0072_{0.0070}^{0.0073}$\\
        espaloma & AM1-BCC & $0.8656_{0.8225}^{0.9131}$ & $1.1398_{1.0715}^{1.2332}$ &
        \multicolumn{2}{c}{reference}\\
        espaloma & espaloma &
        $0.9596_{0.9100}^{1.0101}$ &
        $1.2216_{1.1529}^{1.2766}$ &
        $0.0072_{0.0071}^{0.0073}$ &
        $0.0072_{0.0070}^{0.0073}$
         \\
        \hline
        \multirow{3}{*}{Joint Model, Loss=}

        & 
        $\operatorname{MSE}(
        U_\mathtt{qm}(\mathbf{x}), \hat{U}_\mathtt{valence}(\mathbf{x}; \Phi_\mathtt{NN})
        + \hat{U}_\mathtt{charge}(\mathbf{x}; \Phi_\mathtt{NN})
        + U_\mathtt{LJ}(\mathbf{x})
        )
        $
        & $0.8646_{0.8186}^{0.9180}$ 
        & $1.0839_{1.0462}^{1.1228}$
        & $0.3218_{0.3173}^{0.3258}$
        & $0.3230_{0.3184}^{0.3290}$ \\
        
        & 
        $\operatorname{MSE}(
        U_\mathtt{qm}(\mathbf{x}), \hat{U}_\mathtt{valence}(\mathbf{x}; \Phi_\mathtt{NN})
        + U_\mathtt{AM1-BCC}(\mathbf{x})
        + U_\mathtt{LJ}(\mathbf{x})
        )
        + \operatorname{MSE}(q, \hat{q}(\Phi_\mathbf{NN}))
        $
        & $0.8336_{0.7882}^{0.8877}$ & $1.0987_{1.0318}^{1.1697}$ & $0.0075_{0.0073}^{0.0077}$ & $0.0076_{0.0075}^{0.0077}$\\
        
        & 
        $\operatorname{MSE}(
        U_\mathtt{qm}(\mathbf{x}), \hat{U}_\mathtt{valence}(\mathbf{x}; \Phi_\mathtt{NN})
        + \hat{U}_\mathtt{charge}(\mathbf{x}; \Phi_\mathtt{NN})
        + U_\mathtt{LJ}(\mathbf{x})
        )
        + \operatorname{MSE}(q, \hat{q}(\Phi_\mathbf{NN}))
        $
        & $0.9912_{0.9531}^{1.0443}$
        & $1.2921_{1.2403}^{1.3323}$
        & $0.0138_{0.0136}^{0.0140}$
        & $0.0137_{0.0135}^{0.0139}$
        \\

        \hline

    \end{tabular}}
    \caption{
    \textbf{Espaloma can jointly predict partial charges through a simple charge-equilibration model that learns electronegativity and hardness parameters.}
    Using the OpenFF PhAlkEthOH dataset (244,036 snapshots over 7,408 molecules containing only carbon, oxygen, and hydrogen atoms), we compare the performance of indepdently trained and jointly trained espaloma models in reproducing snapshot potential energies and AM1-BCC partial charges.
    In the first two rows, we train two independent models to predict the energy of conformations and charges of atoms;
    in the third row, we use a joint model where the latent embedding is shared between these two tasks.
    As a reference, the RMSE between AM1-BCC charges assigned by two different chemoinformatics toolkits---Ambertools 21~\cite{wang2004development} and OpenEye Toolkit---are $0.0126_{0.0124}^{0.0129}$ (e).
    }
    \label{tab:joint_charge}
\end{table}

\paragraph{A simple charge-equilibration model can use learned physical parameters}

Instead, predicting the parameters of a simple physical topological charge-equilibration model~\cite{rappe1991charge,gilson2003fast} can produce geometry-independent partial charges capable of reproducing charges derived from quantum chemical electrostatic potential fits~\cite{wang2019graph}.
Note that, unlike \citet{wang2019graph}, here we fit AM1-BCC charges rather than higher level of quantum mechanics theory due to their high cost.
Specifically, we use our atom latent representation to instead predict the first- and second-order derivatives of a pseudopotential energy $E$ with respect to the partial atomic charge $q_i$ on atom $i$:
\begin{equation}
        e_i \equiv \frac{\partial E}{\partial q_i},
        s_i  \equiv \frac{\partial E^2}{\partial^2 q_i}.
    \end{equation} 
Here, the \emph{electronegativity} $e_i$ quantifies the desire for an atom to take up negative charge, while the \emph{hardness} $s_i$ quantifies the resistance to gaining or losing too much charge.
A module is added to Stage III of espaloma to predict the chemical environment adapted $(e_i, s_i)$ parameters for each atom from the latent atom embeddings.
    
The partial charges for all atoms can then be obtained by minimizing the second-order Taylor expansion of the potential pseudoenergy contributed by atomic charges:  
\begin{equation}\label{q}
    \{ \hat{q}_i\} = \underset{q_i}{\mathrm{argmin}}\sum_i \hat{e}_i q_i + \frac{1}{2}\hat{s}_i q_i^2,
\end{equation}
subject to
\begin{equation}\label{constraint}
\sum_i \hat{q}_i = \sum_i q_i = Q,
\end{equation}
where $Q$ is the total (net) charge of the molecule. 

Using Lagrange multipliers, the solution to \ref{q} can be given analytically by:
\begin{equation}
    \hat{q}_i = - e_i s_i^{-1} + s_i^{-1}\frac{Q + \sum_i e_i s_i^{-1}}{\sum_j s_j ^{-1}},
\end{equation}whose Jacobian and Hessian are trivially easy to calculate. As a result, the prediction of $\{\hat{e}_i, \hat{s}_i\}$ could be optimized end-to-end using backpropagation.

\paragraph{A learned charge model predicts AM1-BCC charges to better accuracy than the difference between different implementations of AM1-BCC}

When predicting the partial charges independently, we observe that the RMSE error on the test set ($0.0072_{0.0071}^{0.0073}$ e), is smaller than the difference between the discrepancy between AM1-BCC charges assigned by two popular chemoinformatics toolkits, Ambertools 21~\cite{amber2020} and OpenEye Toolkit ($0.0126_{0.0124}^{0.0129}$).
To translate this error into energy scale, we pair this charge module, which we termed espaloma charge in Table~\ref{tab:joint_charge}, with either openff-1.2.0 or espaloma valence force field and observe that there is only a slight decrease in the total energy performance (within confidence interval).

\paragraph{Espaloma can generate fast and accurate partial charges and valence parameters simultaneously}

Next, we integrated this approach into an espaloma model where parameters of the charge equilibration model $\{e_i, s_i\}$ and the bonded (bond, angle, and torsion) parameters are optimized jointly.
The resulting model is trained by augmenting the loss function to include a term that penalizes the deviation from AM1-BCC partial charges for the molecules in the training set.
On the one hand, one can calculate the Coulomb energy term using this predicted set of charges and incorporate this directly into the energy MSE loss function (shown in the first row in the second half of Table~\ref{tab:joint_charge}).
This approach, although it maintains a relatively accurate energy prediction, leads to a large charge RMSE, since no reference charge is provided.
On the other hand, we can penalize the derivation from the reference charges by adding an MSE loss on the charges with a tunable weight as a hyperparameter (second row), which we tune on the validation set to be $1e-3$.
This setting results in satisfactory performance in both energy and charge prediction.
Finally, if we combine both losses (third row), we observe worse performance on test set energy predictions, which could be attributed to the repeated strong regularization on charge parameters.

\section{Espaloma can parameterize biopolymers}
\label{sec:biopolymers}
We have insofar established that espaloma, as a method to construct MM force fields, shows great versatility and flexibility.
In the following sections, we showcase its utility with a model predicting both valence parameters and partial charges trained on OpenFF Gen2 Optimization Dataset as well as PepConf dataset, which we released as `espaloma-0.2.2` with the package.

\begin{figure}[tbp]
    \centering
    \includegraphics[width=\textwidth]{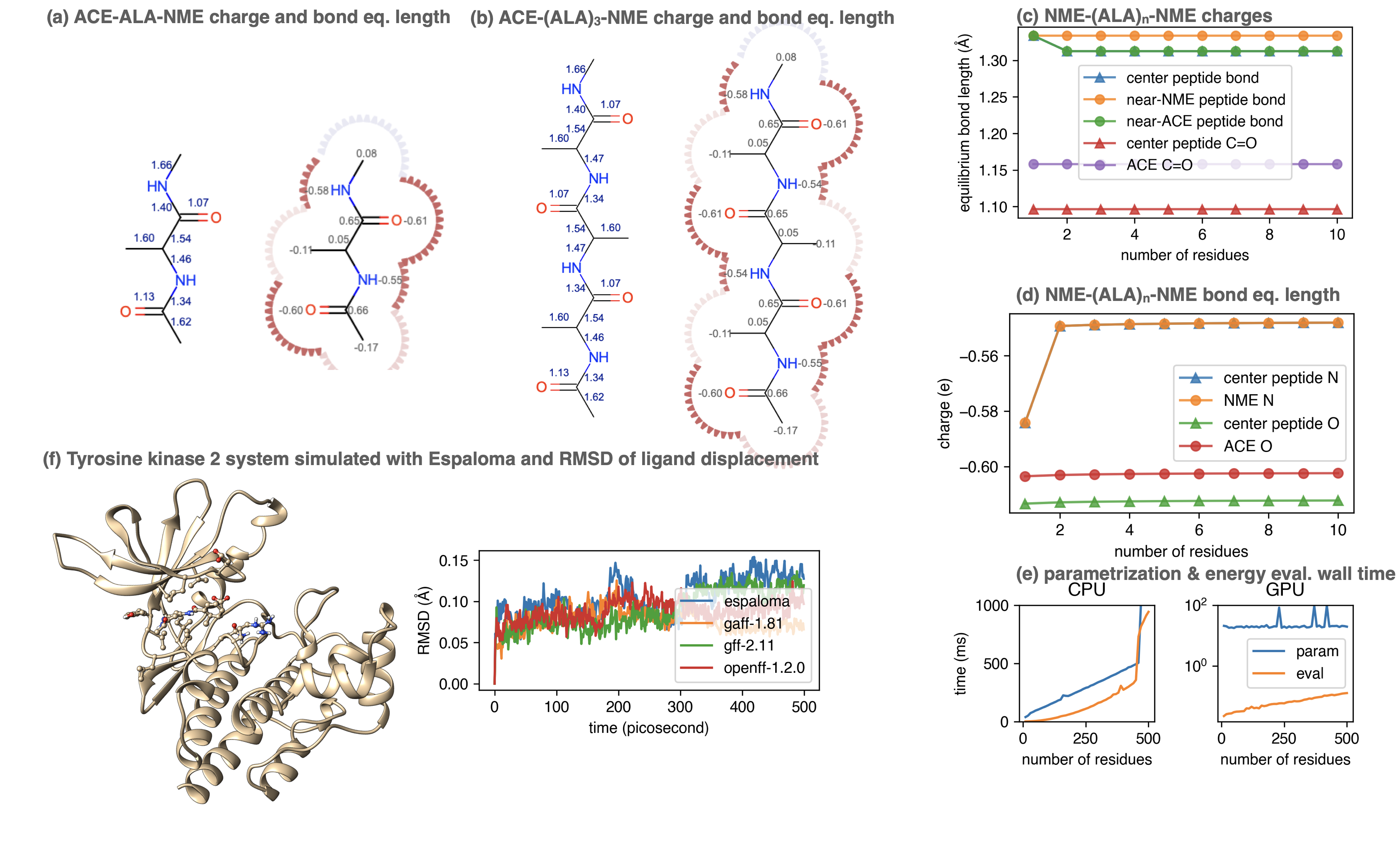}
    \caption{
    \textbf{Espaloma can be used to generate self-consistent force fields for biopolymers and small molecules.}
    (a) Espaloma-assigned equilibrium heavy atom bond lengths (left, in angstrom) and heavy atom partial charges (right, in elemental charge unit) for ACE-(ALA)$_1$-NME.
    (b) Espaloma-assigned equilibrium heavy atom bond lengths (left) and heavy atom partial charges (right) for ACE-(ALA)$_3$-NME, which shows minimal deviation from ACE-(ALA)$_1$-NME and consistent parameters for amino acid residues.
    (c) Partial atomic charges of selected atoms shown for ACE-(ALA)$_n$-NME for $n=1,...,10$.
    (d) Equilibrium bond lengths of selected bonds for ACE-(ALA)$_n$-NME for $n=1,...,10$.
    (e) CPU (left) and GPU (right) parameter assignment (blue curve) and energy evaluation on OpenMM~7.5.1~\cite{eastman2012accelerating} (orange curve) wall times.
    (f) Espaloma simultaneously parametrizes marco- and small molecule in protein-ligand system.
      Left: Tyrosine kinase 2 system parametrized by espaloma and miminized and equilibrated with TIP3P water model~\cite{wang2014building} and counterions.
      Right: Root mean-squared displacement (RMSD) of ligand w.r.t. the initial position, in systems parametrized by espaloma and traditional force fields.
    }
    \label{fig:my_label}
    \label{fig:alan}
\end{figure}

The speed and flexibility of graph convolutional networks allows espaloma to parameterize even very large biopolymers, treating them as (large) small molecules in a graph-theoretical manner.
While graph neural networks perceive nonlocal aspects of the chemical environment around each atom, the limited number of rounds of message passing ensures stability of the resulting parameters when parameterizing systems that consist of repeating residues, like proteins and nucleic acids.

To demonstrate this, we considered the simple polypeptide system ACE-ALA$_n$-NME, consisting of $n$ alanine residues terminally capped by acetyl- and N-methyl amide capping groups.
Using the joint charge and valence term espaloma model, we assigned parameters to ACE-ALA$_n$-NME systems with $n=1,2,...,500$, showing illustrative parameters in Figure~\ref{fig:alan}.
Espaloma stably assigns parameters to the interior residues of peptides even as they increase in length, with parameters of the central residue unchanged after $n > 3$.
This pleasantly resembles the behavior of traditional residue template based protein force fields, even though no templates are used within espaloma's parameter assignment process.

\paragraph{Espaloma can generate self-consistent valence parameters and partial charges for large biopolymers in less than a second} 

Despite its use of a sophisticated graph net machine learning model, the wall time required to parameterize large proteins scales linearly with respect to the number of residues (and hence system size) on a CPU (Figure~\ref{fig:alan}, lower right).
On a GPU, the wall clock time needed to parameterize systems of this size stays roughly constant (due to overhead in executing models on the GPU) at less than 100 microseconds.
Since espaloma applies standard molecular mechanics force fields, the energy evaluation times for an Espaloma-generated force field are identical to traditional force fields.


\section{Espaloma can produce self-consistent biopolymer and small molecule force fields that result in stable simulations}
\label{sec:protein-ligand}


\begin{figure}
    \centering
    \includegraphics[width=\textwidth]{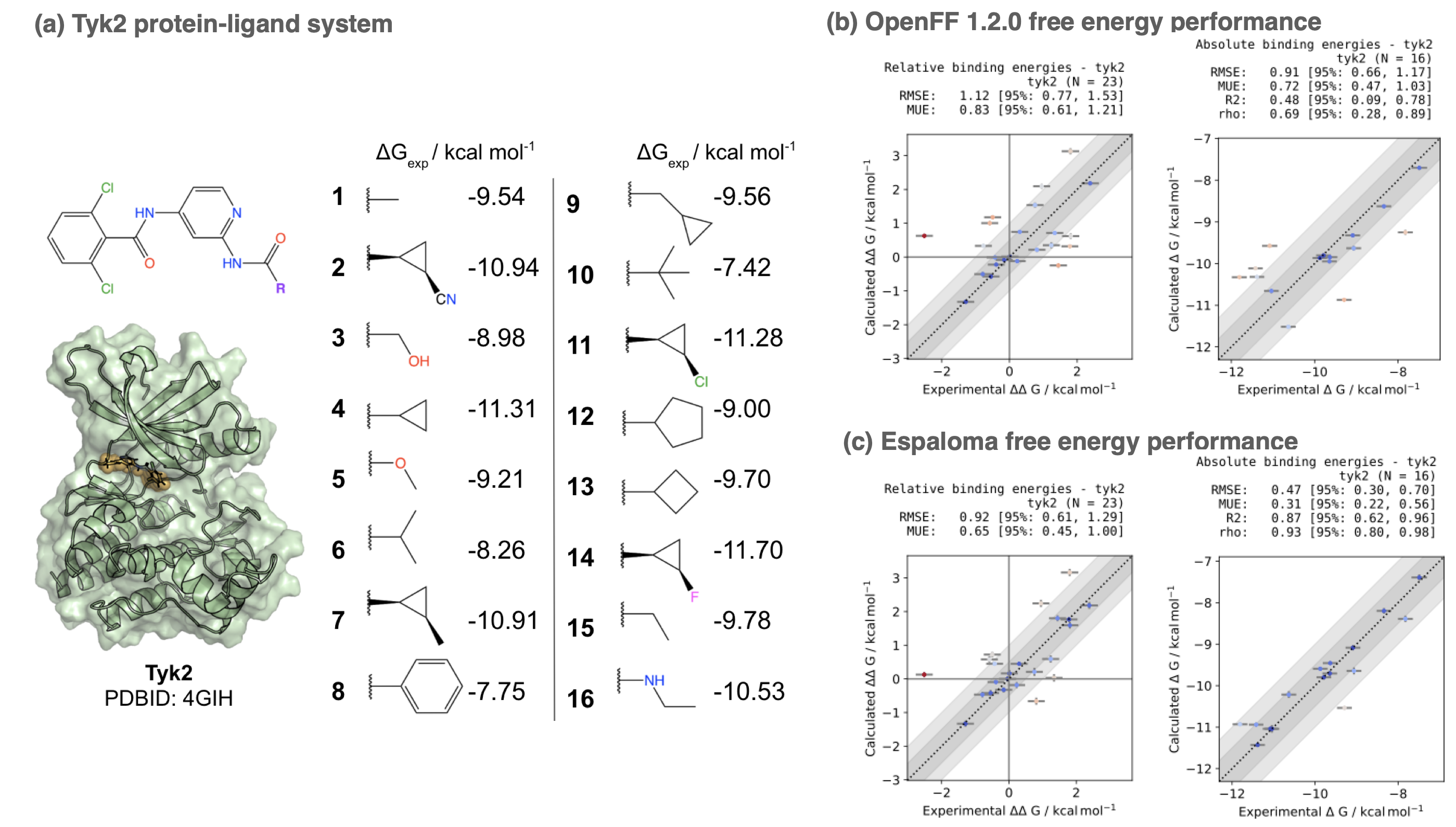}
    \caption{
    {\bf Espaloma small molecule parameters can be used for accurate protein-ligand alchemical free energy calculations.}
    (a) The Tyk2 congeneric ligand benchmark series taken from the Schr\"odinger JACS benchmark set~\cite{wang2015accurate} is challenging for both commercial force fields (OPLS2.1 achieves a $\Delta$G RMSE of 0.93±0.12~kcal~mol$^{-1}$~\cite{wang2015accurate}) and public force fields (GAFF~1.8 achieves a $\Delta$G RMSE of 1.13~kcal~mol$^{-1}$, and $\Delta \Delta$G RMSE of 1.27~kcal~mol$^{-1}$~\cite{song2019using}). 
    We show the X-ray structure used for all free energy calculations as well as 2D structures of all ligands in the benchmark set, along with their experimental binding free energies.
    This congeneric series from~\cite{wang2015accurate} was selected from \cite{liang2013lead} where experimental errors in $K_i$ are reported to have $\delta K_i / K_i < 0.3$, yielding $\delta \Delta G \approx$~0.18~kcal~mol$^{-1}$ and $\delta \Delta \Delta G \approx$~0.25~kcal~mol$^{-1}$.
    Here, we used the perses~\cite{perses} relative free energy calculation tool, based on OpenMM~\cite{eastman2017openmm}, to assess the accuracy of espaloma on this dataset.
    (b) The Open Force Field ("Parsley") openff-1.2.0 small molecule force field achieves an absolute free energy ($\Delta G$) RMSE of 0.91 [95\% CI: 0.66, 1.17] kcal/mol on this set.
    (c) The espaloma-0.2.2 model for predicting valence parameters and partial charges---trained jointly on the same OpenFF Gen2 Optimization dataset used for openff-1.2.0 as well as the PepConf dataset to reproduce quantum chemical energies and AM1-BCC charges---achieves a lower error of 0.48 [95\% CI: 0.30, 0.73] on this set, despite having never been trained on any molecules in this set.
    }
    \label{fig:tyk2-free-energy}
\end{figure}

Traditionally, in a protein-ligand system, separate (but hopefully compatible) force fields and charge models have been assigned to small molecules (which are treated as independent entities parameterized holistically) and proteins (which are treated as collections of templated residues parameterized piecemeal)~\cite{amber2020}.
This practice both has the potential to allow significant inconsistencies while also introducing significant complexity in parameterizing heterogeneous systems.

Using the joint espaloma model trained on both the "OpenFF Gen 2 Optimization" small molecule and "PepConf" peptide quantum chemical datasets (Section~\ref{sec:qm-fitting})~\cite{qiu2020development, prasad2019pepconf}), we can apply a consistent set of parameters to both protein and small molecule components of a kinase:inhibitor system.
Figure~\ref{fig:alan}(f) shows the ligand heavy-atom RMSD after aligning on protein heavy atoms for 0.5~ns trajectories of the Tyk2:inhibitor system from the Alchemical Best Practices Benchmark Set 1.0~\cite{mey2020best}.
It is readily apparent that the espaloma-derived parameters lead to trajectories that are comparably stable to simulations that utilize the Amber ff14SB protein force field~\cite{maier2015ff14sb} with GAFF~1.81, GAFF~2.11~\cite{wang2004development,wang2006automatic}, or OpenFF~1.2.0~\cite{openff-1.2.0} small molecule force fields.
All systems are explicitly solvated with a 9~\AA  buffer around the protein with TIP3P water~\cite{jorgensen1983comparison} and use the Joung and Cheatham monovalent counterion parameters~\cite{joung2008determination} to model a neutral system with 300~mM NaCl salt.

Additionally, the espaloma model also provides sufficient coverage to model more complex and heterogeneous protein-ligand convalent conjugates, which was highly non trivial in traditional force fields where protein and ligand are parametrized separately.
We provide a detailed study of this capability in Appendix Section~\ref{sec:covalent-ligands}.


\section{Espaloma small molecule parameters and charges provide accuracy improvements in alchemical freeenergy calculations}

To assess whether the small molecule parameters and charges generated by espaloma achieve competitive performance to traditional force fields, we used the perses 0.9.5 relative alchemical free energy calculation infrastructure~\cite{perses} (based on OpenMM 7.7~\cite{eastman2017openmm} and openmmtools 0.21.2~\cite{Chodera:openmmtools:2021}) to compare performance on the Tyk2 kinase:inhibitor benchmark set from the Schrodinger JACS benchmark set~\cite{wang2015accurate} as curated by the OpenFF protein-ligand benchmark 0.2.0~\cite{protein-ligand-benchmark}.
In order to assess the impact of espaloma small molecule parameters and charges in isolation, we used the Amber ff14SB protein force field~\cite{maier2015ff14sb}, and performed simulations with either OpenFF~1.2.0 (\textsc{openff-1.2.0}) or the espaloma Joint model trained on OpenFF Gen2 Optimization and PepConf datasets (\textsc{espaloma-0.2.2}) available through the openmmforcefields 0.11.0 package~\cite{openmmforcefields}.
Notably, none of the ligands appearing in this set appear in the training set for either force field.
All systems were explicitly solvated with a 9~\AA  buffer around the protein with TIP3P water~\cite{jorgensen1983comparison} and use the Joung and Cheatham monovalent counterion parameters~\cite{joung2008determination} to model a neutral system with 300~mM NaCl salt.
The same transformation network provided in the OpenFF protein-ligand benchmark set was used to compute alchemical transformations, and absolute free energies up to an additive constant were estimated from a least-squares estimation strategy~\cite{xu2019optimal} as implemented in the OpenFF arsenic package~\cite{arsenic}.
Both experimental and calculated absolute free energies were shifted to their respective means before computing statistics, as in \citep{wang2015accurate}.

\FIG{tyk2-free-energy} shows a comparison of both relative ($\Delta \Delta G$) and absolute ($\Delta G$) free energy error statistics.
While the OpenFF~1.2.0 force field achieves an impressive RMSE of 0.91$_{0.66}^{1.17}$~kcal/mol, using espaloma valence and charge parameters improves the accuracy to 0.47$_{0.30}^{0.70}$~kcal/mol.
Additionally, the Spearman $\rho$ correlation coefficient improves from 0.69$_{0.28}^{0.89}$ (OpenFF~1.2.0) to 0.93$_{0.80}^{0.98}$ (espaloma-0.2.2).
While more extensive benchmarking is necessary to esablish the generality of these improvements, this represents a first demonstration that performance can be on par with, if not superior to, traditionally constructed force fields.

\section{Discussion}
\label{sec:discussion}

Here, we have demonstrated that graph neural networks not only have the capacity to reproduce legacy atom type classification, but they are sufficiently expressive to fit a traditional molecular mechanics force field and generalize it to new molecules, as well as learn entirely new force fields directly from quantum chemical energies and experimental measurements.
The neural framework presented here also affords the modularity to easily experiment with the inclusion of additional potential energy terms, functional forms, or parameter classes, while making it easy to rapidly refit the entire force field afterwards.

%
%
\paragraph{Espaloma enables a wide variety of applications}

Espaloma enables a wide variety of applications in the realm of molecular simulation:
While many force field packages use complex, difficult to maintain, non-portable custom typing engines~\cite{wang2006automatic,vanommeslaeghe2012automation-1,vanommeslaeghe2012automation-2,yesselman2012match}, simply generating examples is sufficient to train espaloma to reproduce this typing, translating it into a model that is easy to extend by providing more quantum chemical training data.
Some force fields have traditionally been typed by hand, making them difficult to automate~\cite{yesselman2012match}; espaloma can in principle learn to generalize from these examples, provided care is taken to avoid overfitting during training.
As we have shown here, espaloma also provides a convenient way to rapidly build new force fields directly from quantum chemical data.

\paragraph{Modern machine learning frameworks offer flexibility in fitting potentials}

The flexibility afforded by modern machine learning frameworks solves a long-standing problem in molecular simulation in which it is extremely difficult to assess whether a new functional form might lead to significant benefits in modeling multiple properties of interest.
While efforts such as the Open Force Field Initiative aim to streamline the process of refitting force fields~\cite{qiu2020development}, the ease of refitting models in machine learning frameworks makes it extremely easy to experiment with new functional forms: Modern automatic differentiation in these frameworks means that only the potential need be implemented, and gradients are automatically computed.

This enables a wide variety of exploration:
Simple improvements could be widely implemented in current molecular simulation packages including adjusting the 1-4 Lennard-Jones and electrostatics scaling parameters, producing 1-4 interaction parameters that override Lennard-Jones combining rules, exploring different Lennard-Jones combining rules~\cite{waldman1993new}, changing the van der Waals treatment to alternative functional forms (such as Buckingham exp-6~\cite{toennies1973validity} or Halgren potentials~\cite{halgren1992representation}), and refitting force fields for various non-bonded treatments (such as PME~\cite{darden1993particle} and reaction field electrostatics~\cite{kubincova2020reaction}).
Many simulation packages provide support for Class~II molecular mechanics force fields~\cite{maple1994derivation,hwang1994derivation}, which include additional coupling terms that can drastically reduce errors in modeling quantum chemical energies at essentially no meaningful impact on cost due to the $\mathcal{O}(N)$ number of these terms; simple extensions to espaloma's architecture can easily predict the parameters for these coupling terms from additional symmetry-preserving features.

More radical potential explorations could involve assessing different algebraic functional forms---modern simulation packages such as OpenMM have the ability to automatically differentiate and compile symbolic algebraic expressions to produce optimized force kernels for simulation on fast GPUs~\cite{eastman2012accelerating,eastman2017openmm}.
Excitingly, the simplicity of incorporating a new generation of quantum machine learning (QML) potentials~\cite{von2018quantum}---such as ANI~\cite{smith2017ani,devereux2020extending} and SchNet~\cite{schutt2018schnet}---means that it will be easy to explore hybrid potentials that combine the flexibility of QML potentials at short range with the accuracy of physical forces at long range~\cite{unke2019physnet}.

\paragraph{Espaloma can enable modular loss functions and regularization}

The ease at which the loss function can be augmented with additional terms enables the addition of other classes of loss terms to the loss function.
For example, one of the molecules considered in the Tyk2:inhibitor system included a cyano group which proved to be slightly unstable with hydrogen mass repartitioning at 4 fs timesteps.
The loss function could either be augmented to regularize parameters to increase stability (penalizing short vibrational periods) or to include other data classes (such as Hessians and/or torsion drive data) to improve fits to particular aspects.
While this will require tuning of the weighting of different loss classes, these parameters can be selected automatically via cross-validation strategies.

\paragraph{Espaloma can enable Bayesian force field parameterization and model uncertainty quantification}

While much of the history of molecular simulation has focused on quantifying the impact of statistical uncertainty~\cite{grossfield2009quantifying,chodera2016simple,grossfield2018best}, critical studies over the last decade~\cite{Cailliez:TheJournalofChemicalPhysics:2011, Cailliez:J.Comput.Chem.:2014,Angelikopoulos:TheJournalofChemicalPhysics:2012, Hadjidoukas:JournalofComputationalPhysics:2015,Cooke:BiophysicalJournal:2008,Zhou::2017, Rizzi:MultiscaleModel.Simul.:2012,Wu:Phil.Trans.R.Soc.A:2016,Kulakova:ArXiv170508533Phys.Stat:2017,Patrone:ArXiv180102483Phys.:2018} have improved our ability to quantify and propagate predictive uncertainty in molecular mechanics force fields by quantifying contributions from model uncertainty---which is frequently the major source of predictive uncertainty in applications of interest.
While most attention has been focused on the \textit{continuous parameters} of the force field model with fixed model form, some progress has been made in discrete model selection among candidate model forms~\cite{Wu:TheJournalofChemicalPhysics:2016, Messerly:TheJournalofChemicalPhysics:2017,Madin:ArXiv210507863Phys.Stat:2021}.

It remains an open problem to rigorously quantify uncertainty in other important parts of the model definition---especially in the definitions of atom-types.
These ``chemical perception'' definitions can involve very large spaces of discrete choices, and crucially influence the behavior of a generalizable molecular mechanics model~\cite{mobley2018escaping,Zanette:J.Chem.TheoryComput.:2019}.

An important benefit of the present approach is that it reduces the mixed continuous-discrete task of ``being Bayesian about atom-types'' to the more familiar task of ``being Bayesian about neural network weights.''
Bayesian treatment of neural networks---while also intractable---has been the focus of productive study and methodological innovation for decades~\cite{Neal::1996}.

We anticipate that Bayesian extensions of this work will enable more comprehensive treatment of predictive uncertainty in molecular mechanics force fields.

\paragraph{Ensuring full chemical equivalence is nontrivial}

In the current experiments, espaloma used a set of atom features (one-hot encoded element, hybridization, aromaticity, formal charge, and membership in rings of various sizes) easily computed using a cheminformatics toolkit; no bond features were used (see {\bf Detailed Methods}).
While this provided excellent performance, the non-uniqueness of formal charge assignment (obvious in molecules such as guanidinium where resonance forms locate the formal charge on different atoms) does not guarantee the assigned parameters will respect chemical equivalence (a form of invariance) in cases where these atom properties are not unique.
Ensuring full chemical equivalence would require modifications to this strategy, such as omission of non-unique features (which may require additional data or pre-training to learn equivalent chemical information), averaging of the output of one or more stages over equivalent resonance forms, or architectures such as transformers that more fully encode chemical equivalence. 




\section{Detailed Methods}
\label{sec:experimental-details}

\subsection{Code and Parameter Availability}
\label{sec:code-availability}

The Python code used to produce the results discussed in this paper is distributed open source under MIT license [\texttt{https://github.com/choderalab/espaloma}].
Core dependencies include PyTorch 1.9.1~\cite{paszke2017automatic}, Deep Graph Library 0.6.0~\cite{wang2019deep}, the Open Force Field Toolkit 0.10.0~\cite{mobley2018open, jeff_wagner_2020_4057038}, and OpenMM 7.7.0~\cite{eastman2017openmm}.

{\bf Describe how espaloma can be used in OpenMM via openmmtools, and describe which model is available as espaloma-0.0.2}

\subsection{Datasets}
\label{sec:datasets}

The typed ZINC validation subset distributed with parm@Frosst~\cite{parm_frosst} was used in atom typing classification experiments (Section~\ref{sec:stage1}).

For MM fitting experiments (Section~\ref{sec:mm_fitting}), we employed molecules the PhAlkEthOH dataset~\cite{PhAlkEthOH-1.0}, parametrized with GAFF-1.81~\cite{doi:10.1002/jcc.20035} using Antechamber~\cite{wang2004development,wang2006automatic} from AmberTools21, and generated molecular dynamics (MD) snapshots with annotated energies according to the procedure detailed below (Section~\ref{sec:md-simulation-details}).
We filtered out molecules with a gap between minimum and maximum energy larger than 0.1 Hartree (62.5 kcal/mol).

For QM fitting experiments (Section~\ref{sec:qm-fitting}), datasets hosted on QCArchive~\cite{smith2020molssi} are used.
We filter out snapshots with energies more than 0.1 Hartree (62.5 kcal/mol) higher than the minima.
Within all datasets, we randomly select training, test, and validation sets with 80:10:10 partitions.

\subsection{Machine learning experimental details}
\label{sec:ml-details}

The input features of the atoms included the one-hot encoded element, as well as the hybridization, aromaticity, (various sized-) ring membership, and formal charge thereof, assigned using the OpenEye Toolkit (OpenEye Scientific Software).

All models are trained with 5000 epochs with the Adam optimizer~\cite{kingma2014adam};
early stopping was used to select the epoch with lowest validation set loss.

Hyperparameters, namely choices of graph neural network layer architectures (GIN~\cite{xu2018powerful}, GCN~\cite{DBLP:journals/corr/KipfW16}, GraphSAGE~\cite{hamilton2017inductive}, SGConv~\cite{DBLP:journals/corr/abs-1902-07153}), depth of graph neural network and of Janossy pooling network (3, 4, 5, 6), activation functions (ReLU, sigmoid, tanh), learning rates (1e-3, 1e-4, 1e-5), and per-layer units (16, 32, 64, 128, 256, 512) were briefly optimized with a grid search using validation sets on the MM fitting experiment.
As a result, we use three 128-units GraphSAGE~\cite{hamilton2017inductive} layers with ReLU activation function for stage I and four 128-units feed-forward layers with ReLU activation for stage II and III.
Reported metrics: $R^2$: the coefficient of determination, RMSE: root mean square error, MAPE: mean absolute percentage error; note that the MAPE results we report is not multiplied by 100, and therefore denotes the fractional error.
The annotated 95\% confidence intervals are calculated by bootstrapping the test set 1000 times to account for finite-size effects in the composition of the test set.

\subsection{Molecular dynamics simulation details}
\label{sec:md-simulation-details}
High-temperature MD simulations described in Section \ref{sec:mm_fitting} were initialized using RDKit's default conformer generator followed by energy minimization in OpenMM 7.5, with initial velocities assigned randomly to the target temperature.
Vacuum trajectories were simulated without constraints using \texttt{LangevinIntegrator} from OpenMM~\cite{eastman2017openmm} using a temperature of 500~K, collision rate of 1/picosecond, and a timestep of 1~fs.
500 samples (5 ns) are collected with 10000 steps (10 ps) between each sample.

\subsection{Alchemical free energy calculations}
\label{sec:alchemical-free-energy-calculations}
We used the perses 0.9.5 relative alchemical free energy calculation infrastructure~\cite{perses} (based on OpenMM 7.7~\cite{eastman2017openmm} and openmmtools 0.21.2~\cite{Chodera:openmmtools:2021}) to compare performance on the Tyk2 kinase:inhibitor benchmark set from the Schrodinger JACS benchmark set~\cite{wang2015accurate} as curated by the OpenFF protein-ligand benchmark 0.2.0~\cite{protein-ligand-benchmark}.
In order to assess the impact of espaloma small molecule parameters and charges in isolation, we used the Amber ff14SB protein force field~\cite{maier2015ff14sb}, and performed simulations with either OpenFF~1.2.0 (\textsc{openff-1.2.0}) or the espaloma Joint model trained on OpenFF Gen2 Optimization and PepConf datasets (\textsc{espaloma-0.2.2}) available through the openmmforcefields 0.11.0 package~\cite{openmmforcefields}.
Notably, none of the ligands appearing in this set appear in the training set for either force field.
All systems were explicitly solvated with a 9~\AA  buffer around the protein with TIP3P water~\cite{jorgensen1983comparison} and use the Joung and Cheatham monovalent counterion parameters~\cite{joung2008determination} to model a neutral system with 300~mM NaCl salt.
The same transformation network provided in the OpenFF protein-ligand benchmark set was used to compute alchemical transformations, and absolute free energies up to an additive constant were estimated from a least-squares estimation strategy~\cite{xu2019optimal} as implemented in the OpenFF arsenic package~\cite{arsenic}.
Both experimental and calculated absolute free energies were shifted to their respective means before computing statistics, as in \citep{wang2015accurate}.

Alchemical free energy calculations used replica exchange among Hamiltonians with Gibbs sampling complete mixing exchanges each iteration~\cite{chodera2011replica}, simulating 5~ns/replica with 1~ps between exchange attempts.
12 alchemical states were used.
Simulations were conducted at 300~K and 1~atm using a Monte Carlo Barostat and Langevin BAOAB integrator~\cite{leimkuhler2016efficient} with bonds to hydrogen constrained, a collision rate of 1/ps, 4 fs timestep, and heavy hydrogen masses.
Atom mappings were generated from the provided geometries in the benchmark set, mapping atoms that were within 0.2Å and subsequently correcting the maps to be valid with the \texttt{use\_given\_geometries} functionality of perses.

\bibliography{main, implicit_solvent_refs, forcefield_uq_refs, implicit_solvent_refs_openff}

\section{Acknowledgments}
\label{sec:acknowledgments}

The authors wish to thank Yutong Zhao (\href{https://orcid.org/0000-0003-0090-7801}{0000-0003-0090-7801}) for helpful feedback while troubleshooting implementations of molecular mechanics models automatic differentiation packages, David L.~Mobley (\href{https://orcid.org/0000-0002-1083-5533}{0000-0002-1083-5533}) for feedback about issues with improper torsion models and other important considerations, Joshua T.~Horton (\href{https://orcid.org/0000-0001-8694-7200}{0000-0001-8694-7200}) for insightful comments especially regarding fitting and assessing torsions.
The authors wish to thank Christian L.\ Mueller (\href{https://orcid.org/0000-0002-3821-7083}{0000-0002-3821-7083}).
We thank OpenEye Scientific Software to provide us with the free academic license.

\section{Funding}
\label{sec:funding}

YW acknowledges support from NSF CHI-1904822 and the Sloan Kettering Institute.
JF acknowledges support from NSF CHE-1738979 and the Sloan Kettering Institute.
JDC acknowledges support from NIH grant P30~CA008748, NIH grant R01~GM121505, NIH grant R01~GM132386, NSF~CHI-1904822, and the Sloan Kettering Institute.

\section{Disclosures}
\label{sec:disclosures}

YW is among the co-founders and equity holders of Uli, Inc. and Uli (Shenzhen) Technology Co.\ Ltd. 

JDC is a current member of the Scientific Advisory Board of OpenEye Scientific Software, Redesign Science, Ventus Therapeutics, and Interline Therapeutics, and has equity interests in Redesign Science and Interline Therapeutics.
The Chodera laboratory receives or has received funding from multiple sources, including the National Institutes of Health, the National Science Foundation, the Parker Institute for Cancer Immunotherapy, Relay Therapeutics, Entasis Therapeutics, Silicon Therapeutics, EMD Serono (Merck KGaA), AstraZeneca, Vir Biotechnology, Bayer, XtalPi, Interline Therapeutics, the Molecular Sciences Software Institute, the Starr Cancer Consortium, the Open Force Field Consortium, Cycle for Survival, a Louis V. Gerstner Young Investigator Award, and the Sloan Kettering Institute.
A complete funding history for the Chodera lab can be found at \url{http://choderalab.org/funding}.

\section{Disclaimers}
\label{sec:disclaimers}

The content is solely the responsibility of the authors and does not necessarily represent the official views of the National Institutes of Health.

\section{Author Contributions}
Conceptualization: JF, YW, JDC;
Data Curation: YW, JF, JEH;
Formal Analysis: YW;
Funding Acquisition: JDC;
Investigation: YW, JF;
Methodology: YW, JF;
Project Administration: JDC;
Resources: JDC;
Software: YW, JF, BK, DR, IZ, IP, MH;
Supervision: JDC;
Visualization: YW;
Writing -- Original Draft: YW;
Writing -- Review \& Editing: YW, JDC, JF, BK, JEH.


\title{Appendix: End-to-End differentiable construction of molecular mechanics force fields}
\newpage
\appendix
\maketitle

\section{A graph theoretic view of Class I molecular mechanics force fields}
\label{sec:molecular-mechanics-forcefields}

Consider a molecular graph $\mathcal{G}$ where atoms map to vertices $\mathcal{V}$ and covalent bonds map to edges $\mathcal{E}$.
In a class I molecular mechanics force field~\cite{maple1994derivation,hwang1994derivation,maple1994derivation3,peng1997derivation,maple1998derivation,dauber2019biomolecular,hagler2019force}, the parameters $\Phi_\mathtt{FF}$ assigned to a molecule graph $\mathcal{G}$ define how the total potential energy of a conformation $\mathbf{x} \in \mathbb{R} ^ {\rvert \mathcal{G} \rvert * 3}$ is computed from independent bond, angle, torsion, and nonbonded energy terms given the complete set of molecular mechanics parameters $\Phi_\mathrm{FF}$. 
\begin{eqnarray}
\label{mm_graph}
    U_\text{MM}(\mathbf{x}; \mathcal{G}, \Phi_\mathtt{FF}) 
    &= \sum\limits_{(v_i, v_j) \in \mathcal{G}_\text{bond}} & U_\text{bond}\left( r(\mathbf{x}; v_i, v_j); K_r(\Phi_\mathtt{FF}; v_i, v_j), r_0(\Phi_\mathtt{FF}; v_i, v_j) \right) \nonumber \\
    &+ \sum\limits_{(v_i, v_j, v_k) \in \mathcal{G}_\text{angle}} & U_\text{angle}\left( \theta(\mathbf{x}; v_i, v_j, v_k); K_\theta(\Phi_\mathtt{FF}; v_i, v_j, v_k), \theta_0(\Phi_\mathtt{FF}; v_i, v_j, v_k) \right) \nonumber \\
    &+ \sum\limits_{(v_i, v_j, v_k, v_l) \in \mathcal{G}_\text{torsion}} & U_\text{torsion}\left( \phi(\mathbf{x}; v_i, v_j, v_k, v_l) ; \{K_{\phi, n}(\Phi_\mathtt{FF}; v_i, v_j, v_k, v_l)\}_{n=1}^{n_\text{max}}, \phi_0(\Phi_\mathtt{FF}; v_i, v_j, v_k, v_l) \right) \nonumber \\
    &+ \sum\limits_{(v_i, v_j) \in \mathcal{G}_\text{Coulomb}} & U_\text{Coulomb} \left( r(\mathbf{x}; v_i, v_j) ; q(\Phi_\mathtt{FF}; v_i), q(\Phi_\mathtt{FF}; v_j) \right) \nonumber \\
    &+ \sum\limits_{(v_i, v_j) \in \mathcal{G}_\text{van der Waals}} & U_\text{van der Waals} \left( r(\mathbf{x}; v_i, v_j) ; \sigma(\Phi_\mathtt{FF}; v_i, v_j), \epsilon(\Phi_\mathtt{FF}; v_i, v_j) \right)
\label{eq: u_mm}
\end{eqnarray}
Here, the sets $\mathcal{G}_\text{bond}, \mathcal{G}_\text{angle}, \mathcal{G}_\text{torsion}$ denote the duples, triples, and quadruples of bonded atoms (vertices) in $\mathcal{G}$, while $\mathcal{G}_\text{Coulomb}$ and $\mathcal{G}_\text{van der Waals}$ denotes the set of atom (vertex) pairs separated by at least \textit{three} edges, since interactions separated by fewer edges are generally excluded.
$\epsilon_0$ denotes the vacuum electric permittivity.
The potential terms depend on distances $r(\mathbf{x}; v_i, v_j)$, angles $\theta(\mathbf{x}; v_i, v_j, v_k)$, and torsions (dihedral angles) $\phi(\mathbf{x}; v_i, v_j, v_k, v_l)$ measured for the corresponding atoms from the positions vector $\mathbf{x}$.
The various parameter functions---bond force constant $K_r$ and equilibrium distance $r_0$, angle force constant $K_\theta$ and equilibrium angle $\theta_0$, periodic torsion term barrier height $K_{\phi, n}$ and phase $\phi_0$ for periodicity $n$, partial charge $q$, Lennard-Jones radius $\sigma$ and well depth $\epsilon$---extract the parameters corresponding to specific sets of atoms (vertices) from the vector of molecular mechanics parameters $\Phi_\mathrm{FF}$ to compute that contribution to the total potential energy.


The individual potential energy terms are computed by functions that generally take simple harmonic or periodic forms with respect to bond lengths and angles in terms of the parameters for those specific interactions:
\begin{eqnarray}
&U_\text{bond}(r; K_r, r_0) &= \frac{K_r}{2} (r - r_0)^2 \\
&U_\text{angle}(\theta; K_\theta, \theta_0) &= \frac{K_\theta}{2} (\theta - \theta_0)^2\\
&U_\text{torsion}(\phi; \{K_{\phi, n}\}, \phi_0) &= \sum\limits_{n=1}^{n_\text{max}} K_{\phi, n}\left[1 + \cos(n \phi)\right] \\
&U_\text{Coulomb}(r; q_i, q_j) &= \frac{1}{4 \pi \epsilon_0} \frac{q_i \, q_j}{r}  \\
&U_\text{van der Waals}(r; \sigma, \epsilon) &= 4 \epsilon \left[ \left (\frac{\sigma}{r} \right)^{12} - \left(\frac{\sigma}{r}\right)^{6} \right]
\label{eq: detailed_u_mm}
\end{eqnarray}
where $K_r$ and $K_\theta$ denote force constants for bonds and angles, $r_0$ and $\theta_0$ denote equilibrium bond lengths and angles, $K_{\phi, n}$ denotes a torsion energy factor (which can be positive or negative) for periodicity $n$, $\epsilon_0$ is the permittivity of free space, and $q_i$ denote partial charges.
For force fields like AMBER~\cite{maier2015ff14sb} and CHARMM~\cite{best2012optimization}, the effective Lennard-Jones well depth $\epsilon_{v_i, v_j}$ and radius $\sigma_{v_i, v_j}$ parameters for an interacting pair of atoms are computed from atomic parameters for atoms $v_i$ and $v_j$ using Lorentz-Berthelot combining rules~\cite{delhommelle2001inadequacy},
\begin{eqnarray}
&\sigma(\Phi_\mathtt{FF}; v_i, v_j) &= \frac{1}{2} (\sigma_{v_i} + \sigma_{v_j}) \\
&\epsilon(\Phi_\mathtt{FF}; v_i, v_j) &= \sqrt{\epsilon_{v_i} \epsilon_{v_j}}
\end{eqnarray}
though alternative combining rules~\cite{waldman1993new} or even pair-specific parameters~\cite{baker2010accurate} are also possible.

Typically, atom type based force field parameterization engines~\cite{mobley2018escaping} such as those used in AMBER~\cite{maier2015ff14sb} or CHARMM~\cite{best2012optimization} assign parameters based on templates (for biopolymer residues or solvents) or through chemical perception algorithms (as in the case of GAFF~\cite{wang2004development,wang2006automatic} or CGenFF~\cite{vanommeslaeghe2012automation-1,vanommeslaeghe2012automation-2}). 
More recently, an approach to assigning atom, bond, angle, and torsion parameters directly based on the common SMARTS/SMIRKS chemical perception language was introduced, which bypasses the need to define parameter classes in terms of atom types, but still retains discrete interaction types~\cite{mobley2018escaping}.
Automatically fitting these type definitions still represents an intractable mixed discrete-continuous optimization problem.
In the next section, we will show how a continuous, differentiable model can assign parameters directly, without the use of discrete atom or interaction types.


\section{A brief introduction to graph neural networks}
\label{sec:introduction-to-graph-nets}

In the context of molecular machine learning, molecules are modelled as undirected graphs of bonded atoms, where each atom and bond can carry attributes reflecting their chemical nature from which complex chemical features can be learned. 
If we write this as a tuple of three sets,
\begin{equation}
\mathcal{G} = \{ \mathcal{V, E, U}\}
\end{equation} 
Here, $\mathcal{V}$ is the set of the vertices (or nodes) (atoms), $\mathcal{E}$ the set of edges (bonds), and $\mathcal{U} = \{ \mathbf{u}\}$ the universal (global) attribute. 

In a graph neural network ($\operatorname{GN}$) a set of functions (with learnable parameters $\Phi_{\operatorname{NN}}$) govern the three stages used in both training and inference of a graph neural network: \textit{initialization}, \textit{propagation}, and \textit{readout}. 
Following one of the most general description of the message-passing procedure in the propagation stage in \cite{battaglia2018relational}, we briefly review the message-passing steps in graph neural networks, where node attribuves $\mathbf{v}$, edge attributes $\mathbf{e}$, and global attributes $\mathbf{u}$ are updated according to:
    \begin{align}
    \mathbf{e}_k^{(t+1)} &= \phi^e(\mathbf{e}_k^{(t)}, \sum_{i \in \mathcal{N}^e_k}\mathbf{v}_i, \mathbf{u}^{(t)}),
    &&\text{edge update}\\
    \bar{\mathbf{e}}_i^{(t+1)} &=\rho^{e\rightarrow v}(E_i^{(t+1)}),
    &&\text{edge-to-node aggregate}\\
    \mathbf{v}_i^{(t+1)} &= \phi^v(\bar{\mathbf{e}}_i^{(t+1)}, \mathbf{v}_i^{(t)}, \mathbf{u}^{(t)}),
    &&\text{node update}\\
    \bar{\mathbf{e}}^{(t+1)} &= \rho^{e \rightarrow u}(E^{(t+1)}),
    &&\text{edge-to-global aggregate}\\
    \bar{\mathbf{u}}^{(t+1)} &= \rho^{v \rightarrow u}(V^{(t)}),
    &&\text{node-to-global aggregate}\\
    \mathbf{u}^{(t+1)} &= \phi^u(\bar{\mathbf{e}}^{(t+1)}, \bar{\mathbf{v}}^{(t+1)}, \mathbf{u}^{(t)}),
    &&\text{global update}
    \end{align}
    where $E_i=\{ \mathbf{e}_k, k\in \mathcal{N}_i^v\}$ is the set of attributes of edges connected to a specific node, $E_i = \{ e_k, k \in 1, 2, ..., N^e\}$ is the set of attributes of all edges, $V$ is the set of attributes of all nodes, and $\mathcal{N}^v$ and $\mathcal{N}^e$ denote the set of indices of entities connected to a certain node or a certain edge, respectively. 
    $\phi^e$, $\phi^v$, and $\phi^u$ are update functions that take the \textit{environment} of the an entity as input and update the attribute of the entity, which could be stateful or not; 
    $\rho^{e \rightarrow v}$, $\rho^{e \rightarrow u}$, and $\rho^{v \rightarrow u}$ are aggregation functions that aggregate the attributes of multiple entities into an \textit{aggregated} attribute which shares the same dimension with each entity. 
    Note that it is common that the edges do not hold attribute but only pass messages onto neighboring nodes. 
    For all models we survey here, edge-to-global update does not apply and global attribute does not present until the readout stage, when a sum function is applied to form the global representation ($\mathbf{u} = \sum_i \mathbf{v}_i$). 
    We review the specifics of the graph neural network architectures that we considered here in the Appendix (Table~\ref{gn}).
    
\subsection{Training and inference}
\label{training_and_inference}

While traditional force fields based on discrete atom types are only differentiable with respect to the molecular mechanics parameters they assign, our model is fully differentiable in \emph{all} model parameters $\Phi_{\operatorname{NN}}$ that govern both the assignment of continuous atom embeddings $h_{v}$ (which replace discrete atom types) and subsequent assignment of MM parameters $\Phi_{\text{FF}}$.
We can therefore use gradient-based optimization to tune all of these parameters to fit arbitrary differentiable target functions, such as energies and forces of snapshots of molecules, equilibrium physical property measurements, fluctuation properties, or experimentally measured free energy differences.
In this way, straightforward gradient-based optimization simultaneously learns the continuous equivalent of atom typing simultaneous with parameter assignment strategies, effectively solving the intractable traditional force field parmaeterization problem of mixed continuous-discrete optimization by replacing it with a fully continuous optimization problem.

\subsubsection{A linear basis facilitates bond and angle parameter optimization}
\label{sec: linear_basis}
\begin{figure}[h]
    \centering
    \includegraphics[width=0.6\textwidth]{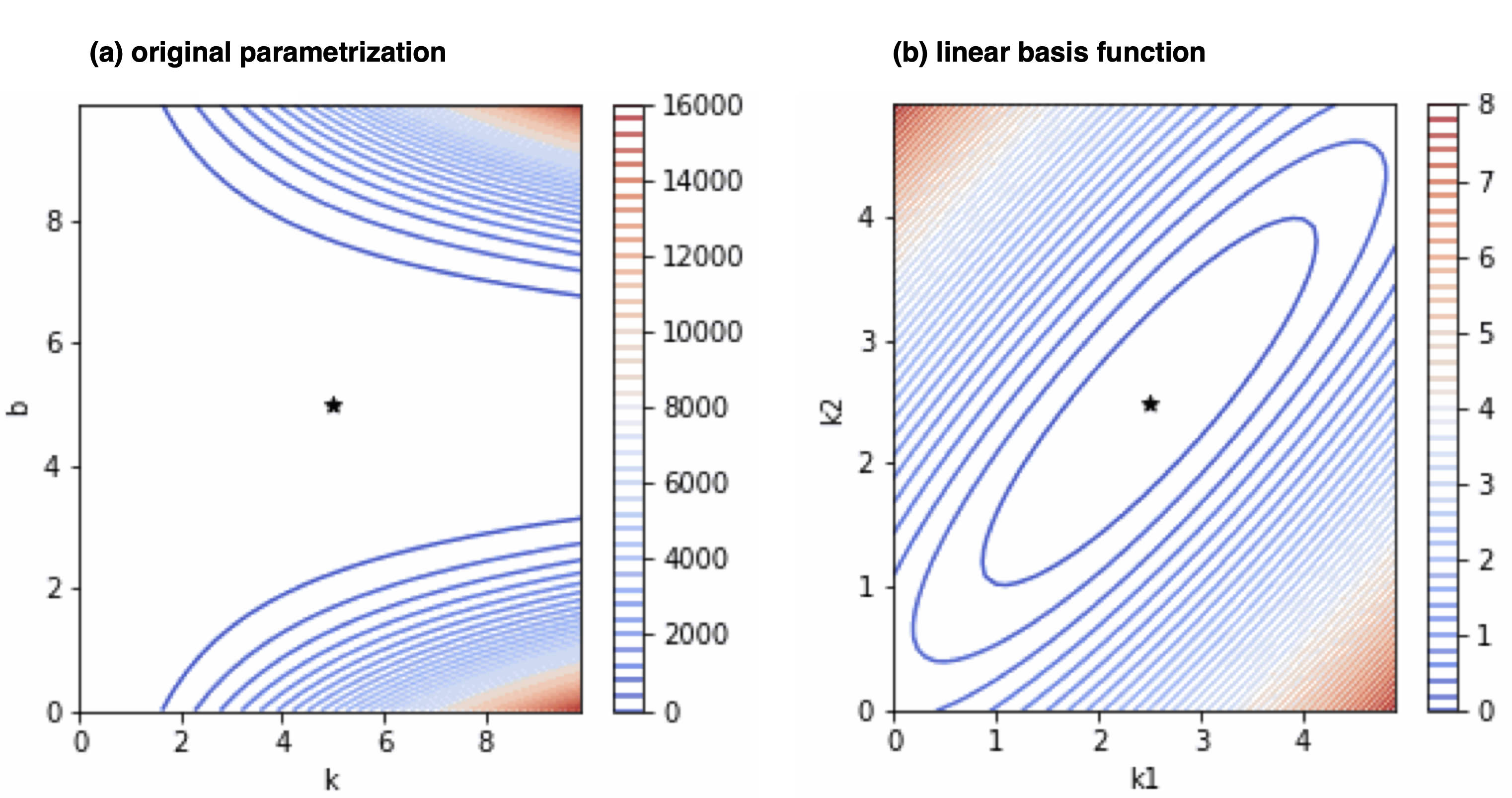}
    \caption{
    \textbf{Linear basis parameterization of bonds and angles facilitates robust optimization.}
    A one-dimensional harmonic oscillator system ($u = \frac{K_0}{2} (x - b_0) ^ 2$) was simulated with $K_0$ = 5, $b_0$ = 5 as reference parameters.
    Training set conformations $x$ were sampled from a uniform distribution with $x \sim \operatorname{Uniform}(4, 6)$.
    The loss function (squared energy error) as a function of original MM parameters ($K$, $b$) shown in (a) is clearly more difficult to optimize due to the large difference in gradient magnitudes between $K$ and $b$ parameters, while the linear basis parametrization (Section~\ref{sec: linear_basis}) showing the loss as a function of transformed $(K_1, K_2)$ in (b) is much simpler to optimize. 
    }
    \label{fig:linear_basis}
\end{figure}

The harmonic functional form (Equation~\ref{eq: detailed_u_mm}),
\begin{equation}
u(x; K, b) = \frac{K}{2} (x-b)^2 ,
\label{eq: harmonic}
\end{equation}
is frequently used in MM potentials for modeling bond and angle energies.
Although one can directly optimize force constants $K$ and equilibrium bond or angle values $b$, this empirically result in significant difficulties in training (Figure~\ref{fig:linear_basis}).
This can be seen through the loss function and its gradient for a simple toy model consisting of a single bond and reference parameters ($K_0, b_0$):
\begin{eqnarray}
L(K, b) &=& \sum_{n=1}^N [ u(x_n; K, b) - u(x_n; K_0, b_0) ]^2 \\
&=& \sum_{n=1}^N [ \frac{K}{2} (x_n - b)^2 - \frac{K_0}{2} (x_n - b_0)^2 ]^2
\end{eqnarray}
Qualitatively, the loss function landscape for $K$ around $K_0$ is relatively flat, while it is steeply varying for $b$ around $b_0$, frustrating efficient optimization; Figure~\ref{fig:linear_basis}a demonstrates this for a simple toy system.
To circumvent this issue, we use the approach described by \citet{https://doi.org/10.1002/jcc.23897} and translate the harmonic functional form Equation~\ref{eq: harmonic} in which $(K, b_0)$ are optimized instead into a linear combination ($K_1, K_2$) of basis functions is optimized instead:
\begin{equation}
u(x; K_1, K_2) = \underbrace{\frac{K_1}{2} (x-b_1)^2 + \frac{K_2}{2} (x-b_2)^2}_\text{linear combination of basis} + \underbrace{\left(- K_1 b_1^2 - K_2 b_2^2 + \frac{(K_1 b_1 + K_2 b_2)^2}{K_1 + K_2}\right)}_\text{$x$-independent constant}
\end{equation}
where $b_1$ and $b_2$ are two positive constants chosen to be smaller and larger than any possible valid bond or angle---they must satisfying $b_1 < x < b_2 $ for $\forall x$ that might be observed.
Here, Stage~III of the espaloma pipeline would predict $K_1, K_2$, which would be used to compute the MM parameters $(K, b)$ via
\begin{align}
    K &= K_1 + K_2 \\
    b &= \frac{K_1 b_1 + K_2 b_2}{K_1 + K_2}
\end{align}
When training to fit energies, our loss function subtracts the mean from the training set energies and predicted energies for each molecule to remove the arbitrary offset constant introduced by this transformation of variables.

\subsubsection{Training by potential energies}
\label{sec:training-by-potential-energies}

Given a training set of molecules $\{ \mathcal{G}_m\}$, $m = 1, \ldots, M$, with corresponding $n = 1, \ldots, N_m$ conformational snapshots $\mathbf{x}_{m,n} \in \mathbf{R}^{(\mid \mathcal{G}_m \mid \times 3)}$ for each molecule $m$, and reference potential energies $\{ U_\text{ref}(x_{m,n}; ; \mathcal{G}_m) \}$, the model parameters $\Phi_{\operatorname{NN}}$ can be optimized to minimize a measure of deviation between the reference energies and model energies given by the composed force field $\Phi_{\text{FF}, \Theta}$. 
For example, a squared loss function can be used to quadratically penalize deviations:
\begin{equation}
\label{eq:loss-function}
\mathcal{L}(\Phi_{\operatorname{NN}}) =  \sum_{m=1}^M \sum_{n=1}^{N_m} w_{m,n} \, \left[ U_\text{ref}(x_{m,n}; \mathcal{G}_m) - U_{\Phi_{\text{FF}, \Theta}}(\mathbf{x}_{m,n}; \mathcal{G}_m) \right]^2
\end{equation}
In this work, we took the weights $w_{m,n} = 1$, but more sophisticated weighting schemes can be used to emphasize low-energy snapshots where MM potentials are intended to be more accurate when fitting to quantum chemical datasets, as in ForceBalance~\cite{wang2013systematic,wang2014building}.
In a Bayesian context, the loss function in Eq.~\ref{eq:loss-function} would correspond to the negative log likelihood for a normal error model.
Other loss or likelihood function forms may also be useful in ensuring a minority of molecules or geometries for which the MM functional form poorly reproduces quantum chemical energetics do not dominate the fit.

Additional terms can be added to the loss function to regularize the choice of parameters, such as penalizing unphysical choices (such as unphysical magnitudes or parameter regions) or to provide more useful models (such as penalizing bond vibrations with effective periods that are so short they would require extremely small timesteps).

When fitting to quantum chemical energies, an additive offset for each molecule $m$ corresponding to the heat of formation generally cannot be accounted for using standard MM functional forms.
To address this, we subtract the per-molecule mean potential energy for both predicted MM and reference quantum chemical energies in formulating the loss function, with the goal of ensuring that relative conformational energetics are accurately approximated even if the heat of formation cannot be computed.

\subsubsection{Other differentiable objectives}

In quantum chemical calculations, nuclear potential energy gradients are often available at very low cost once the wavefunction has been solved to compute the energy.
To exploit the rich information available in these gradients~\cite{christensen2020role, chmiela2019, sauceda2019, sauceda2020molecular} $\nabla_\mathbf{x} U$, it is possible to incorporate both energies and gradients into the loss function, with automatic differentiation used to compute partial derivatives of the espaloma energy function in terms of both coordinates and parameters.
Additional terms penalizing deviations in more complex quantum chemical properties, such as polarizabilities or vibrational frequencies, could also be incorporated~\cite{wang2013systematic,wang2014building}.
In this work, however, we found that incorporation of gradient information does not significantly enhance the performance of the model.
We leave this topic to future study.

In addition to incorporating configuration-dependent properties, if the goal is to build a complete molecular mechanics force field to model entire physical systems in the condensed phase, additional terms could be added to the objective function to quantify the deviation in physical properties, such as densities, dielectric constants, and enthalpies of vaporization~\cite{horn2004development,wang2013systematic,wang2014building,wang2017building,boulanger2018optimized,simonboothroyd_2020_openff_evaluator}, or even experimental transfer free energies.

\subsubsection{Assigning espaloma parameters to molecules}

In contrast with many quantum machine learning (QML) force fields that use a neural model to compute the energy and gradient for every configuration~\cite{smith2017ani, schutt2017schnet, C7SC04934J}, {\bf espaloma} uses a neural model to assign molecular mechanics parameters \emph{once} at the beginning of a simulation, using only the chemical information about all components in a system.
Once generated, the MM parameters $\Phi_\text{FF}$ for the system can be seamlessly ported to molecular mechanics packages that can exploit accelerated hardware~\cite{eastman2017openmm, harvey2009acemd, salomon2013routine, van2005gromacs} to achieve high accuracy for individual systems at the same speed as traditional force fields.


\section{Graph neural network (GNN) architectures considered in this paper}

We considered the following graph neural network architectures available in DGL~\cite{wang2019deep} in this paper:
\begin{table}[h!]
    \label{table:gn}
    \centering
    \begin{tabular}[\textwidth]{c c c c }
    \hline
        {\bf Model} & 
        {\bf Edge update $\phi^e$} &
        {\bf Edge aggregate $\rho^{e\rightarrow v}$} &
        {\bf Node update $\phi^v$} \\
    \hline
        GCN & Identity & Mean &
        $\operatorname{NN}$\\
        
        EdgeConv & $\operatorname{ReLU}(W_0 (\mathbf{v}_i - \mathbf{v}_j) + W_1 \mathbf{v}_i)$ & Max & Identity \\
        
        GraphSAGE &
        Identity &
        Mean$^*$ & 
        $\operatorname{Normalize}(\operatorname{NN}([\mathbf{v}:\mathbf{e}]))$\\
        
        GIN &
        Identity &
        Sum$^*$ &
        $\operatorname{NN}((1 + \epsilon) \mathbf{v} + \mathbf{e})$
        
    \\\hline
    \end{tabular}
    \caption{\textbf{Summary of representative graph neural network architectures by edge update, edge aggregate, and node update types.} Models analyzed here include: GCN~\cite{DBLP:journals/corr/KipfW16}, EdgeConv~\cite{wang2019dynamic}, GraphSAGE~\cite{hamilton2017inductive}, and GIN~\cite{xu2018powerful}. 
    Other architectures evaluated---TAGCN~\cite{Du:ArXiv171010370CsStat:2018} and SGC~\cite{wu2019simplifying}---involve multi-step propagation, which could be expressed as a combination of these updates and aggregates. \\ *: Multiple aggregation functions studied in the referenced publication.}
    \label{gn}
\end{table}

\section{Code snippets for using espaloma}
\subsection{Designing and training espaloma model}
\begin{lstlisting}[language=Python, numbers=none, caption=Defining and training a modular espaloma model.]
import torch, dgl, espaloma as esp

# retrieve OpenFF Gen2 Optimization Dataset
dataset = esp.data.dataset.GraphDataset.load("gen2").view(batch_size=128)

# define espaloma stage I: graph -> atom latent representation
representation = esp.nn.Sequential(
    layer=esp.nn.layers.dgl_legacy.gn("SAGEConv"), # use SAGEConv implementation in DGL
    config=[128, "relu", 128, "relu", 128, "relu"], # 3 layers, 128 units, ReLU activation
)

# define espaloma stage II and III: 
# atom latent representation -> bond, angle, and torsion representation and parameters
readout = esp.nn.readout.janossy.JanossyPooling(
    in_features=128, config=[128, "relu", 128, "relu", 128, "relu"],
    out_features={              # define modular MM parameters espaloma will assign
        1: {"e": 1, "s": 1}, # atom hardness and electronegativity
        2: {"log_coefficients": 2}, # bond linear combination, enforce positive
        3: {"log_coefficients": 3}, # angle linear combination, enforce positive
        4: {"k": 6}, # torsion barrier heights (can be positive or negative)
    },
)

# compose all three espaloma stages into an end-to-end model
espaloma_model = torch.nn.Sequential(
                 representation, readout, esp.nn.readout.janossy.ExpCoefficients(),
                 esp.mm.geometry.GeometryInGraph(), esp.mm.energy.EnergyInGraph(),
                 esp.nn.readout.charge_equilibrium.ChargeEquilibrium(),
)

# define training metric
metrics = [
    esp.metrics.GraphMetric(
            base_metric=torch.nn.MSELoss(), # use mean-squared error loss
            between=['u', "u_ref"],         # between predicted and QM energies
            level="g", # compare on graph level
    )
    esp.metrics.GraphMetric(
            base_metric=torch.nn.MSELoss(), # use mean-squared error loss
            between=['q', "q_hat"],         # between predicted and reference charges
            level="n1", # compare on node level
    )
]

# fit espaloma model to training data
results = esp.Train(
    ds_tr=dataset, net=espaloma_model, metrics=metrics,
    device=torch.device('cuda:0'), n_epochs=5000,
    optimizer=lambda net: torch.optim.Adam(net.parameters(), 1e-3), # use Adam optimizer
).run()

torch.save(espaloma_model, "espaloma_model.pt") # save model
\end{lstlisting}

\subsection{Deploying espaloma model}

\begin{lstlisting}[language=Python, numbers=none, caption=Using a trained espaloma model to assign parameters to a small molecule.]
# define or load a molecule of interest via the Open Force Field toolkit
from openff.toolkit.topology import Molecule
molecule = Molecule.from_smiles("CN1C=NC2=C1C(=O)N(C(=O)N2C)C")

# create an espaloma Graph object to represent the molecule of interest
import espaloma as esp
molecule_graph = esp.Graph(molecule)

# apply a trained espaloma model to assign parameters
espaloma_model = torch.load("espaloma_model.pt")
espaloma_model(molecule_graph.heterograph)

# create an OpenMM System for the specified molecule
openmm_system = esp.graphs.deploy.openmm_system_from_graph(molecule_graph)
\end{lstlisting}

\section{Close examination of molecular mechanics (MM) fitting experiments}
\label{sec:close-examination}

\paragraph{Ablation study: the role of energy centering, the inclusion of torsion energies, and the inclusion of small rings in datasets in the performance of MM fitting.}
\label{para:ablation}

\begin{table}[htbp]
    \centering
 \resizebox{\textwidth}{!}{%
 \setlength\tabcolsep{1.5pt}
\begin{tabular}{ll|llll|llll|llll}
\toprule
 {\bf R~C~T} & Part &             $U_\mathtt{total}$ &             $U_\mathtt{bond}$ &             $U_\mathtt{angle}$ &       $U_\mathtt{torsion}$ & $U_\mathtt{total}$ & $U_\mathtt{bond}$ & $U_\mathtt{angle}$ & $U_\mathtt{torsion}$ &                      $k_r$ &                      $b_r$ &                 $k_\theta$ &                 $b_\theta$ \\
 
 &
 & \multicolumn{4}{c}{RMSE (kcal/mol)}
 & \multicolumn{4}{c}{Centered RMSE (kcal/mol)}
 & \multicolumn{4}{c}{MAPE}
 \\
 
\midrule
     \multirow{2}{*}{000} &   test &     $0.8723_{0.8139}^{0.9334}$ &    $1.2725_{1.0575}^{1.4465}$ &     $1.2291_{0.9563}^{1.4102}$ &                            &   $0.7829_{0.7025}^{0.8428}$ &  $0.6501_{0.5919}^{0.7308}$ &   $0.4553_{0.4245}^{0.4860}$ &                                & $0.1238_{0.1206}^{0.1265}$ & $0.0026_{0.0025}^{0.0027}$ & $0.0733_{0.0722}^{0.0744}$ & $0.0067_{0.0066}^{0.0068}$ \\
    &   train &     $0.9183_{0.8343}^{0.9889}$ &    $1.3159_{1.0655}^{1.5514}$ &     $1.2844_{1.0052}^{1.5360}$ &                            &   $0.8681_{0.7993}^{0.9330}$ &  $0.7209_{0.6639}^{0.7771}$ &   $0.4815_{0.4431}^{0.5097}$ &                                & $0.1293_{0.1253}^{0.1320}$ & $0.0027_{0.0027}^{0.0028}$ & $0.0726_{0.0712}^{0.0736}$ & $0.0071_{0.0069}^{0.0073}$ \\
    
     \multirow{2}{*}{001} &   test &     $4.4033_{4.1317}^{4.6173}$ & $13.7162_{12.9510}^{14.5515}$ &     $7.9210_{7.3953}^{8.4048}$ & $0.0363_{0.0345}^{0.0384}$ &   $3.5238_{3.3984}^{3.6590}$ &  $3.0956_{2.9908}^{3.2187}$ &   $1.7086_{1.6367}^{1.7619}$ &     $0.0015_{0.0014}^{0.0017}$ & $1.0000_{1.0000}^{1.0000}$ & $0.0720_{0.0701}^{0.0740}$ & $0.4734_{0.4712}^{0.4763}$ & $0.0136_{0.0134}^{0.0138}$ \\
     &   train &     $4.2372_{3.9854}^{4.5129}$ & $13.2695_{12.1669}^{14.2219}$ &     $8.0467_{7.4667}^{8.6800}$ & $0.0353_{0.0330}^{0.0379}$ &   $3.5041_{3.3614}^{3.6277}$ &  $3.0241_{2.9115}^{3.1645}$ &   $1.7425_{1.6659}^{1.8182}$ &     $0.0016_{0.0014}^{0.0017}$ & $1.0000_{1.0000}^{1.0000}$ & $0.0705_{0.0686}^{0.0720}$ & $0.4803_{0.4778}^{0.4831}$ & $0.0137_{0.0135}^{0.0139}$ \\
     \multirow{2}{*}{010$^*$} &   test &     $0.0490_{0.0369}^{0.0585}$ &    $0.0206_{0.0158}^{0.0262}$ &     $0.0429_{0.0362}^{0.0497}$ &                            &   $0.0175_{0.0163}^{0.0187}$ &  $0.0130_{0.0121}^{0.0140}$ &   $0.0117_{0.0105}^{0.0127}$ &                                & $0.0016_{0.0016}^{0.0017}$ & $0.0001_{0.0001}^{0.0001}$ & $0.0020_{0.0019}^{0.0020}$ & $0.0006_{0.0006}^{0.0007}$ \\
    &   train &     $0.0394_{0.0315}^{0.0478}$ &    $0.0174_{0.0142}^{0.0206}$ &     $0.0332_{0.0271}^{0.0394}$ &                            &   $0.0169_{0.0160}^{0.0178}$ &  $0.0128_{0.0118}^{0.0135}$ &   $0.0113_{0.0103}^{0.0121}$ &                                & $0.0016_{0.0016}^{0.0016}$ & $0.0001_{0.0001}^{0.0001}$ & $0.0020_{0.0020}^{0.0021}$ & $0.0006_{0.0005}^{0.0006}$ \\
     \multirow{2}{*}{011}&   test &   $11.0688_{8.8832}^{12.8267}$ &    $0.0240_{0.0204}^{0.0284}$ &     $0.1801_{0.1535}^{0.2122}$ & $0.0174_{0.0145}^{0.0203}$ &   $0.0346_{0.0324}^{0.0369}$ &  $0.0168_{0.0157}^{0.0180}$ &   $0.0270_{0.0226}^{0.0307}$ &     $0.0001_{0.0001}^{0.0001}$ & $0.0020_{0.0020}^{0.0021}$ & $0.0001_{0.0001}^{0.0001}$ & $0.0020_{0.0020}^{0.0021}$ & $0.0009_{0.0008}^{0.0009}$ \\
    &   train &   $11.7819_{9.4086}^{14.0593}$ &    $0.0190_{0.0179}^{0.0203}$ &     $0.1899_{0.1545}^{0.2335}$ & $0.0185_{0.0152}^{0.0216}$ &   $0.0335_{0.0302}^{0.0357}$ &  $0.0161_{0.0150}^{0.0173}$ &   $0.0284_{0.0238}^{0.0321}$ &     $0.0001_{0.0001}^{0.0001}$ & $0.0020_{0.0020}^{0.0021}$ & $0.0001_{0.0001}^{0.0001}$ & $0.0020_{0.0020}^{0.0021}$ & $0.0009_{0.0008}^{0.0009}$ \\
    \multirow{2}{*}{100}&   test &     $0.6237_{0.5459}^{0.7049}$ &    $0.5224_{0.4774}^{0.5707}$ &     $0.6660_{0.5843}^{0.7498}$ &                            &   $0.4970_{0.4483}^{0.5526}$ &  $0.3380_{0.3176}^{0.3614}$ &   $0.4654_{0.4084}^{0.5205}$ &                                & $0.0469_{0.0459}^{0.0479}$ & $0.0012_{0.0012}^{0.0013}$ & $0.0645_{0.0624}^{0.0662}$ & $0.0111_{0.0107}^{0.0117}$ \\
    &   train &     $0.6615_{0.5706}^{0.7584}$ &    $0.6199_{0.5480}^{0.6905}$ &     $0.6939_{0.6107}^{0.7702}$ &                            &   $0.5693_{0.4963}^{0.6494}$ &  $0.3916_{0.3473}^{0.4304}$ &   $0.5293_{0.4489}^{0.6160}$ &                                & $0.0482_{0.0473}^{0.0493}$ & $0.0013_{0.0013}^{0.0013}$ & $0.0659_{0.0637}^{0.0677}$ & $0.0112_{0.0108}^{0.0117}$ \\
     \multirow{2}{*}{101} &   test &     $4.7719_{4.6354}^{4.9088}$ &    $9.5963_{9.1500}^{9.9684}$ &  $23.1972_{22.2922}^{24.2192}$ & $0.0501_{0.0480}^{0.0519}$ &   $4.3643_{4.2759}^{4.4498}$ &  $2.6994_{2.6075}^{2.7756}$ &   $3.6263_{3.5347}^{3.6948}$ &     $0.0020_{0.0019}^{0.0022}$ & $0.8410_{0.8386}^{0.8440}$ & $0.0263_{0.0259}^{0.0267}$ & $1.0000_{1.0000}^{1.0000}$ & $0.9287_{0.9282}^{0.9294}$ \\
    &   train &     $4.8043_{4.6731}^{4.9617}$ &    $9.3807_{8.8160}^{9.8568}$ &  $23.2491_{22.4420}^{24.1563}$ & $0.0495_{0.0481}^{0.0514}$ &   $4.4050_{4.2726}^{4.5169}$ &  $2.6971_{2.6393}^{2.7786}$ &   $3.5734_{3.4918}^{3.6850}$ &     $0.0019_{0.0018}^{0.0021}$ & $0.8404_{0.8371}^{0.8428}$ & $0.0267_{0.0263}^{0.0272}$ & $1.0000_{1.0000}^{1.0000}$ & $0.9281_{0.9275}^{0.9288}$ \\
     \multirow{2}{*}{110} &   test & $89.0006_{70.4894}^{102.2105}$ &    $0.0384_{0.0339}^{0.0440}$ & $88.9808_{64.5673}^{107.0304}$ &                            &   $0.0223_{0.0200}^{0.0257}$ &  $0.0246_{0.0221}^{0.0269}$ &   $0.0278_{0.0247}^{0.0310}$ &                                & $0.0018_{0.0018}^{0.0019}$ & $0.0001_{0.0001}^{0.0001}$ & $0.0033_{0.0032}^{0.0035}$ & $0.0074_{0.0064}^{0.0089}$ \\
     &   train &  $79.6927_{62.4512}^{99.1904}$ &    $0.1155_{0.0351}^{0.1924}$ &  $79.6589_{59.2236}^{98.9583}$ &                            &   $0.0209_{0.0185}^{0.0236}$ &  $0.0815_{0.0243}^{0.1566}$ &   $0.0814_{0.0264}^{0.1351}$ &                                & $0.0018_{0.0017}^{0.0020}$ & $0.0001_{0.0001}^{0.0001}$ & $0.0031_{0.0029}^{0.0033}$ & $0.0062_{0.0051}^{0.0076}$ \\
     \multirow{2}{*}{111}&   test & $99.6155_{75.0240}^{118.2089}$ &    $0.0586_{0.0520}^{0.0661}$ & $91.7758_{71.2381}^{115.4096}$ & $0.0245_{0.0222}^{0.0270}$ &   $0.0297_{0.0262}^{0.0352}$ &  $0.0327_{0.0282}^{0.0371}$ &   $0.0438_{0.0398}^{0.0492}$ &     $0.0000_{0.0000}^{0.0001}$ & $0.0017_{0.0016}^{0.0018}$ & $0.0001_{0.0001}^{0.0001}$ & $0.0034_{0.0032}^{0.0036}$ & $0.0077_{0.0066}^{0.0087}$ \\
    &   train & $90.1890_{70.8691}^{113.6430}$ &    $0.1682_{0.0543}^{0.2789}$ &  $81.6664_{64.1198}^{98.1073}$ & $0.0264_{0.0236}^{0.0290}$ &   $0.0279_{0.0258}^{0.0299}$ &  $0.1015_{0.0311}^{0.1693}$ &   $0.1043_{0.0440}^{0.1679}$ &     $0.0001_{0.0000}^{0.0001}$ & $0.0017_{0.0017}^{0.0018}$ & $0.0001_{0.0001}^{0.0001}$ & $0.0033_{0.0031}^{0.0035}$ & $0.0065_{0.0055}^{0.0075}$ \\
\bottomrule
\end{tabular}}
\caption{
\textbf{Best performance of MM fitting on PhAlkEthOH dataset only when centering is deployed and torsion energy and small rings are excluded.}
A comparison of espaloma MM fitting experiments on PhAlkEthOH dataset with/without fitting torsions, centering energies, and including three- and four-membered rings.
The {\bf R~C~T} column summarizes the experimental conditions:
{\bf R} indicates whether small (3- and 4-membered) rings were included in the dataset (1 small rings included, 0 excluded);
({\bf C}) indicates whether the potential energies for sampled snapshots of each molecule were \emph{centered} by subtracting off the mean energy for that molecule in the loss function (1 centered, 0 uncentered);
({\bf T}) indicates whether torsion energies were included in the fitting (1 included, 0 excluded).
The rest of the experimental setting is identical with Figure~\ref{sec:mm_fitting}.
Force field parameters ($\Phi_\text{FF}$):
$K_r$: bond force constant;
$b_r$: equilibrium bond length;
$K_\theta$: angle force constant;
$b_\theta$: equilibrium angle value.
The root mean squared error (RMSE) and correlation coefficient ($R^2$) between reference and predicted MM energies, as well as mean absolute percentage error (MAPE) (as a fraction rather as a percentage, i.e. not multiplied by 100), and correlation coefficient ($R^2$) between reference and predicted force field parameters. 
The sub- and superscripts report the the 95\% confidence interval of each statistics estimated from 1000 bootstrapped replicates on the molecule set.
*: Same experiment as Figure~\ref{sec:mm_fitting}.}
\label{tab:mm_fitting_comparison}
\end{table}

We closely examined several choices which appear to have significant impact on the ability to produce high-quality generalizable models for the task of learning an MM potential:
Concretely, we conduct an ablation study (Table~\ref{tab:mm_fitting_comparison}) where we repeat the experiment shown in Figure~\ref{sec:mm_fitting} considering several variations of these choices:\\
{\bf Small rings (R):} We consider whether molecules with small (3- or 4-membered) rings are included in the dataset, since these molecules assume geometries have near-degeneracies in equilibrium angles (Figure~\ref{fig:cyclopropane}). \\
{\bf Energy centering (C):} We consider the role of energy centering in the loss function, where the average potential energy over all conformations for each molecule is subtracted for both reference and predicted energies in the loss function, such that only relative conformational energy errors are penalized. 
We evaluate the error both with and without centering.
{\bf Torsion energies (T):} We also consider whether torsion energies are included as well, since recovery of torsion profiles can be challenging due to both the near-degeneracy of parameters and phases. 
Notably, while MM force fields generally specify a minimal set of a few torsion periodicities with phases of $0$ or $\pi$ explicitly specified, we use a formulation where all periodicities $n = 1, 2, \ldots, 6$ are always fit and the sign of the torsion $K$ parameter indicates whether the phase is $0$ ($K > 0$) or $\pi$ ($K < 0$).\\
Experiments for all combinations of these choices are shown in Table~\ref{tab:mm_fitting_comparison}.

We summarize our findings as follow:
First of all, best performance is only achieved when we center the predicted and reference energies in the loss function used in training, effectively training to fit only relative energies, suggesting this makes robust optimization easier. 
The rationalization could be found in Section~\ref{sec: linear_basis}.
We also examined the training trajectory of espaloma on this MM fitting task with and without the centering of predicted and reference energies (Figure~\ref{fig:w_wo_centering}) and noticed that centering alleviates, but not completely remedies the training difficulties of espaloma models.

Secondly, the inclusion of torsion energies deteriorates performance, especially if measured by uncentered RMSE. 
This could be attributed to the fact that we use a slightly different functional form to express torsion energies than reference legacy force fields. 

Lastly, with energy centering, including three- and four-membered rings in our training set results in very large uncentered RMSE without drastically changing the centered RMSE (relative conformational energy errors). 
We study this effect in detail in the next section.

\begin{figure}
    \centering
    \small
    \includegraphics[width=0.9\textwidth]{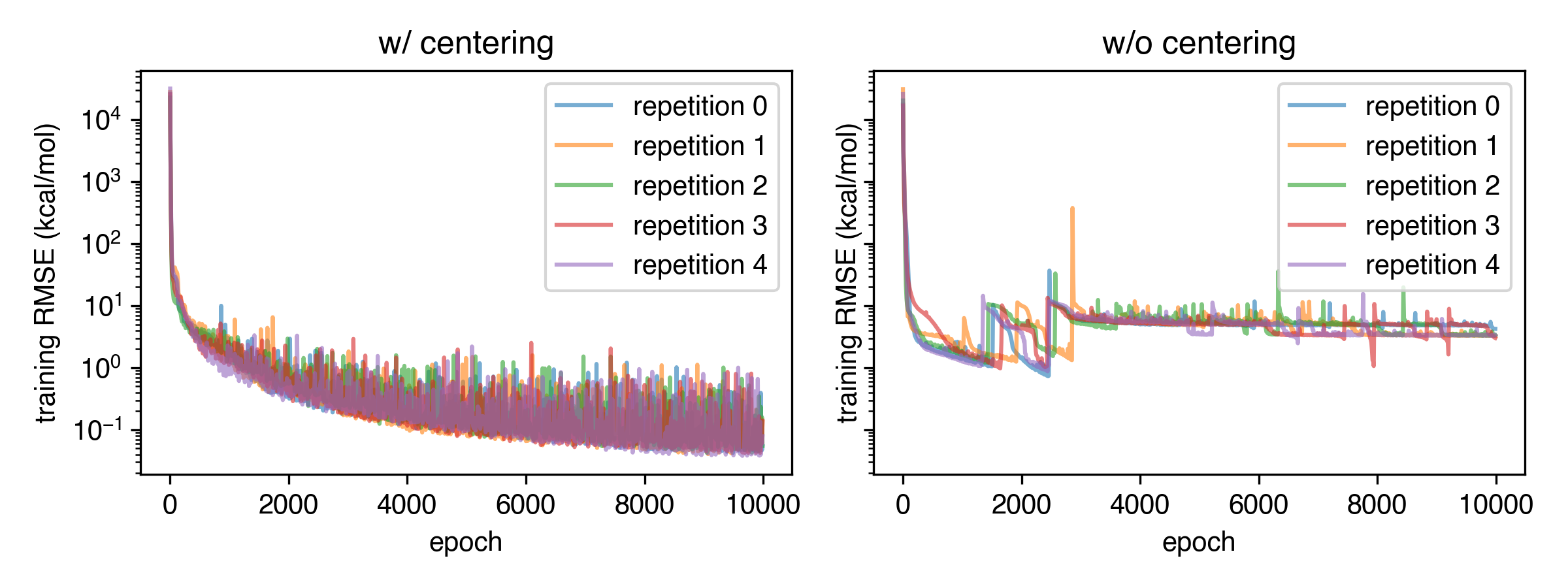}
    \begin{tabular}{c c c c c | c c c c c }
    \hline
    \multicolumn{5}{c}{w/ centering} & \multicolumn{5}{c}{w/o centering}\\
    \hline
    \multirow{2}{*}{Rep.} & \multirow{2}{*}{ES Epoch} & \multicolumn{3}{c}{RMSE at ES (kcal/mol)}
    & \multirow{2}{*}{Rep.} & \multirow{2}{*}{ES Epoch} & \multicolumn{3}{c}{RMSE at ES (kcal/mol)}\\
    &  & Train & Validation & Test 
    & &  & Train & Validation & Test\\
    \hline
0 & 9602 & 0.0456 & 0.0523 & 0.0834 & 0 & 2447 & 0.7331 & 0.6937 & 0.7705 \\
1 & 8938 & 0.0381 & 0.0500 & 0.0815 & 1 & 1902 & 1.2449 & 1.1950 & 1.1910  \\
2 & 9967 & 0.0438 & 0.0503 & 0.0545 & 2 & 1424 & 1.1646 & 1.1351 & 1.1537  \\
3 & 9696 & 0.0411 & 0.0342 & 0.0714 & 3 & 2422 & 0.9849 & 0.9121 & 0.9805  \\
4 & 9879 & 0.0377 & 0.0470 & 0.0714 & 4 & 2446 & 0.8727 & 0.8006 & 0.8460  \\
    \hline
    \end{tabular}
    \caption{
    \textbf{espaloma models are difficult to train, especially when attempting to fit to absolute (rather than relative) conformation energies.}
    We compare the training trajectory of espaloma without torsion energy on PhAlkEthOH dataset excluding three- and four-membered rings with and without centering (experiments 010 and 000 in Table~\ref{tab:mm_fitting_comparison}). 
    Example training trajectories with the same setting as Figure~\ref{sec:mm_fitting}.
    Various repetitions are plotted in different colors.
}
\label{fig:w_wo_centering}
\end{figure}

\paragraph{Case Study: Parametrization of Cyclopropane}
\label{para:cyclopropane}
\begin{figure}
    \centering
    \includegraphics[width=\textwidth]{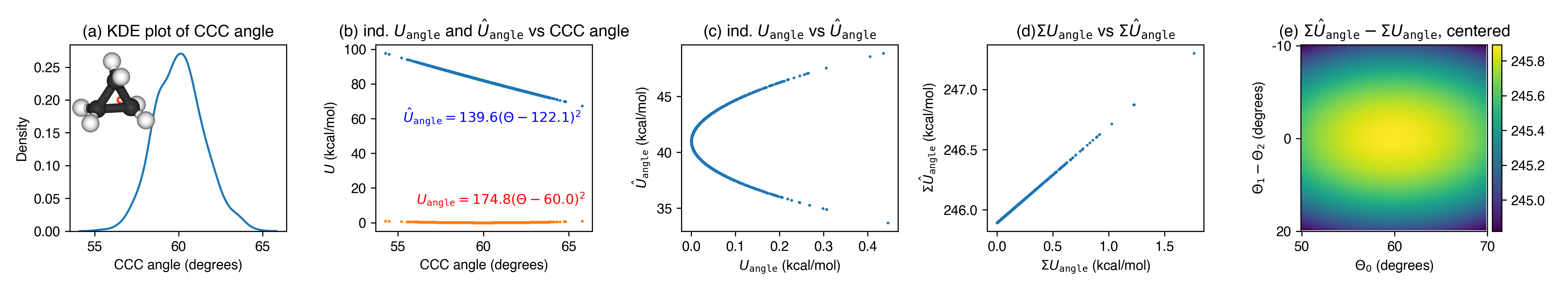}
    \caption{
\textbf{A comparison of reference and espaloma-learned parametrization of cyclopropane.}
    (a): KDE plot of CCC angle in cyclopropane.
    (b): \textbf{Individual} reference and predicted angle energy ($U_\mathtt{angle}$ and $\hat{U}_\mathtt{angle}$, energy of each CCC angle) plotted against CCC angle values in the snapshots generated.
    (c): Predicted plotted against reference \textbf{individual} CCC energy.
    (d): Predicted plotted against reference \textbf{sum} CCC energy.
    (e): Distribution of difference between reference $\Sigma U_\mathtt{angle}$ and predicted $\Sigma\hat{U}_\mathtt{angle}$ \textbf{sum} CCC energy.
    }
    \label{fig:cyclopropane}
\end{figure}

We use the model trained on PhAlkEthOH dataset with the inclusion of small rings but without torsion energies (experiment 110 in Table~\ref{tab:mm_fitting_comparison}) and apply this model to cyclopropane (SMILES: \texttt{C1CC1})
As shown in (b) in  Figure~\ref{fig:cyclopropane}, when trained with centered energy loss function, espaloma assigned equilibrium angle value $\theta_0 \approx 122.1^{\circ}$ to the CCC angle in cyclopropane while the equilibrium angle is $60^{\circ}$ in the reference GAFF-1.81 force field.
With a shift in angle force constant, while this would drastically change the energies of individual CCC angles, the sum of energies of the three CCC angles (whose angle values would sum up to $180^{\circ}$) would remain very close to the reference value up to a constant.
We furthermore compared (in Figure~\ref{fig:cyclopropane}-(e) ) the sum of CCC angle energies in these two parametrizations on a wide range of geometries of cyclopropane with three CCC angles $\theta_0, \theta_1, \theta_2$ satisfying  $ 50^{\circ} < \Theta_0 < 70^{\circ}, -10^{\circ} < \theta_1 - \theta_2 < 10 ^{\circ}$ and noticed that the difference is always a constant with fluctuations within 1 kcal/mol.



\begin{figure}

    \centering
    \includegraphics[width=\textwidth]{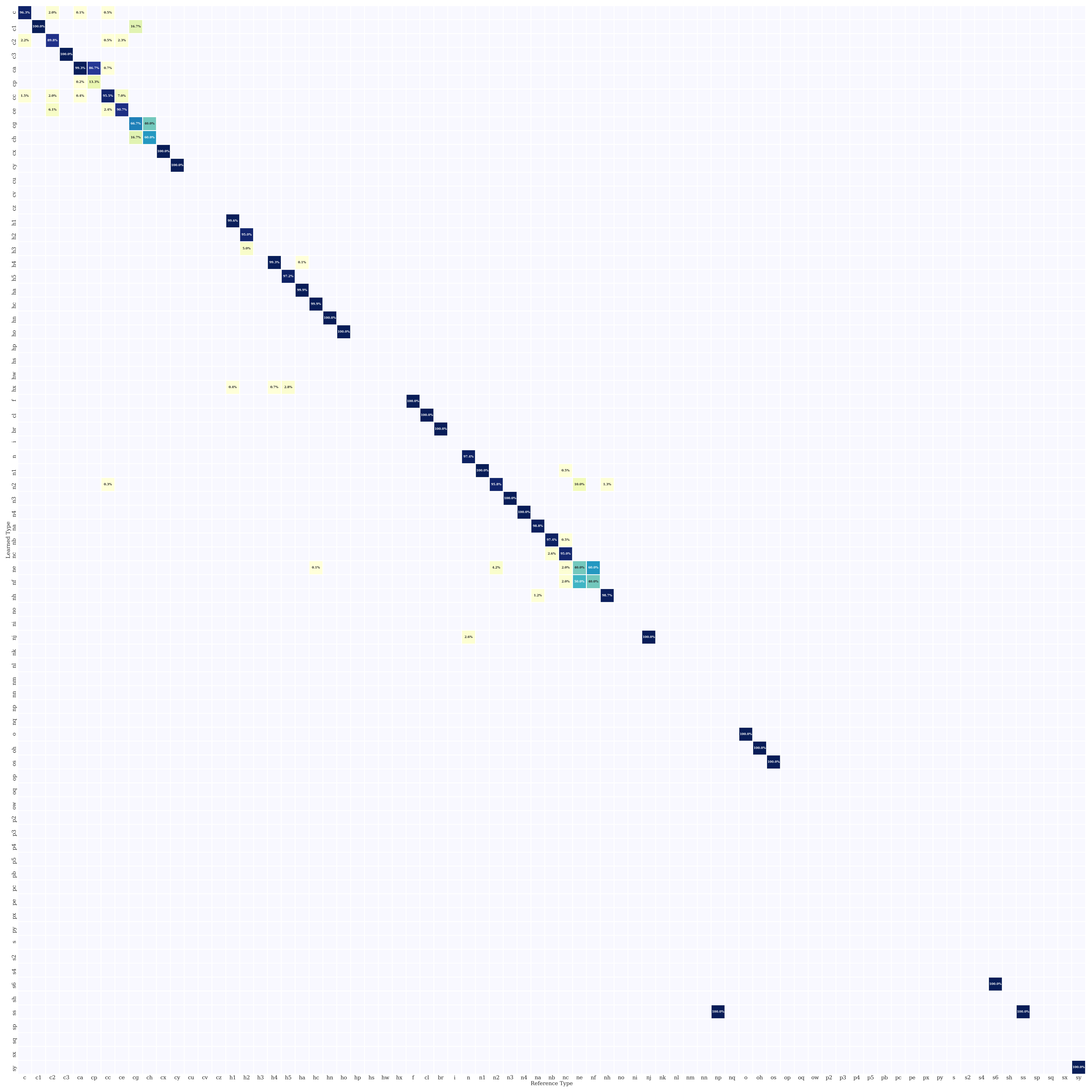}
    \caption{\textbf{Confusion Matrix: Reference vs Learned Atom Types.}
    Continuation of the confusion matrix shown in Figure \ref{fig:atom-typing-accuracy}, to include not just carbon types.
   The blank entries are because the dataset does not cover some of the rare atom types.
   }
    \label{extra-confusion}
\end{figure}


\section{Espaloma can easily parameterize complex heterogeneous biomolecular systems}
\label{sec:covalent-ligands}

\begin{figure}[tbp]
\centering
    \includegraphics[width=0.9\textwidth]{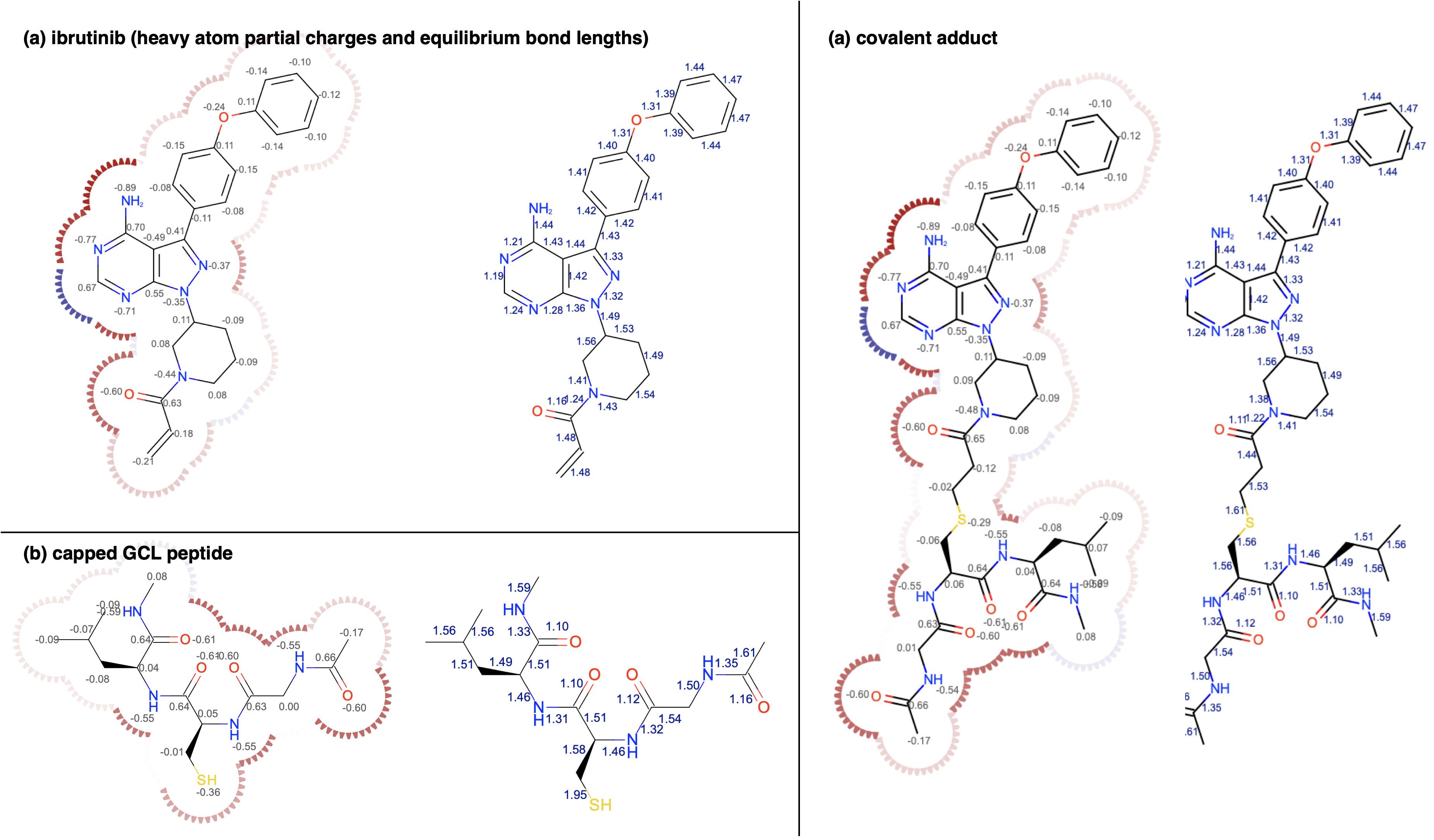}
  \caption{
  \textbf{espaloma parametrizes covalent inhibitor adduct system.}
  Heavy atom bond equilibrium length (in angstrom) and partial charge (in elemental charge) prediction for 
  (a) ibrutinib;
  (b) capped GCL peptide;
  and
  (c) their covalent adduct.}
\label{fig:btk-ibrutinib}
\end{figure}

Because biomolecular systems of interest are often highly heterogeneous environments, it is often practically incredibly difficult to model systems that consist of more than the simplest protein-ligand-solvent combinations.
For example, the latest Amber~20 release~\cite{amber2020} recommends a set of independently developed (but loosely coupled) set of force fields and complex parameterization machinery for proteins~\cite{maier2015ff14sb}, DNA~\cite{galindo2016assessing}, RNA~\cite{perez2007refinement}, water~\cite{jorgensen1983comparison}, monovalent~\cite{joung2008determination,joung2009molecular} and divalent~\cite{li2013rational,li2014taking,li2015parameterization} counterions, lipids~\cite{dickson2014lipid14}, carbohydrates~\cite{kirschner2008glycam06}, glycoconjugates~\cite{demarco2009atomic,demarco2010presentation}, small molecules~\cite{wang2004development,wang2006automatic}, post-translational modifications~\cite{khoury2013forcefield_ptm}, and nucleic acid modifications~\cite{tan2018rna}, collectively representing hundreds of human-years of effort. 
Despite this, even seemingly simple common tasks---such as modeling irreversible covalent inhibitors---represent a mind-bending technical challenge for parameterization~\cite{khan2018covalent}.

Espaloma operates on only the chemical graph of components of the system, which enables it to solve many issues with these legacy approaches:
In generating parameters, there is no practical difference between any of the biopolymer, biomolecule, organic molecule, solvent, ion, or covalently modified species that previously required enormous effort to separately build parameter sets to cover; espaloma will simply interpolate the parameters from observed examples in a continuous manner.
Instead of separately curating distinct datasets and philosophies for parameterizing each of these different classes of chemical species, extending an espaloma model to a new class of chemical species is simply a matter of extending the training datasets of quantum chemical and/or physical property data.
As we have seen, building an elementary model capable of simulating solvated protein-ligand complexes simply required incorporating quantum chemical datasets used to parameterize a small molecule force field alongside quantum chemical data for short peptides.

This espaloma model should also provide sufficient coverage to model more complex protein-ligand covalent conjugates, since the relevant chemistry to model this complex already exists in the training dataset.
To demonstrate that this leads to stable, sensible parameters, we considered the pre-reactive form of the covalent kinase inhibitor ibrutinib and the terminally-blocked form of the Val-Cys-Gly sequence that it reacts with in its target kinase BTK (Figure~\ref{fig:btk-ibrutinib}, left).
Espaloma is able to rapidly parameterize the covalent conjugate of ibrutinib to this BTK target sequence, resulting in minor changes to parameters and partial charges around the covalent warhead, but minimally perturbing the remainder of the system (Figure~\ref{fig:btk-ibrutinib}, right).
Unlike legacy approaches to parameterizing covalently-modified residues, no expensive parameter or charge refinement procedure was needed for this task.

\section{Proof that Janossy pooling is sufficiently expressive}
\label{sec:janossy-works}

Now we prove that such formulation is expressive enough to distinguish bonds consisting of distinct atoms.
An equivalent proof for angles and torsions follows similarly.
\begin{lemma}
\label{lemma:janossy-works}
There exists a neural function $\operatorname{NN}_r$, such that, 
\begin{equation}
\operatorname{NN}_r ([h_{v_i}:h_{v_j}]) + \operatorname{NN}_r ([h_{v_j}:h_{v_i}])
=
\operatorname{NN}_r ([h_{v_k}:h_{v_l}]) + \operatorname{NN}_r ([h_{v_k}:h_{v_l}])\nonumber
\end{equation}
if and only if $h_{v_i} = h_{v_l}$ \textit{and} $h_{v_j} = h_{v_k}$ \textit{or} $h_{v_i} = h_{v_l}$ \textit{and} $h_{v_j} = h_{v_k}$.
\end{lemma}
\begin{proof}
We first prove that there is a \textit{function} $f_r$ satisfying the condition in Lemma  \ref{lemma:janossy-works}.
With finite possible initial conditions and finite rounds of message passing, the possible values of $h_v$ is finite (corresponding to the finite number of unique labels in Weisfeiler-Leman~\cite{weisfeiler1968reduction} test).
We use $c$ to denote the maximum value of the hash of $h_v$:
\begin{equation}
c = \operatorname{max}\{ \operatorname{hash}(h_v)\}
\end{equation}
Thus, one example of such $f_r: \mathbb{R}^{D} \rightarrow \mathbb{N}$ satisfying the condition in Lemma \ref{lemma:janossy-works} is:
\begin{equation}
f_r(h_{v_i}, h_{v_j}) = \begin{cases}
0, \operatorname{hash}(h_{v_i}) \geq \operatorname{hash}(h_{v_j});\\
(c + 1)\operatorname{hash}(h_{v_i}) + \operatorname{hash}(h_{v_i}), 
\operatorname{hash}(h_{v_i}) < \operatorname{hash}(h_{v_j}).
\end{cases}
\end{equation}
Following the universal approximation theorem~\cite{hornik1991approximation, hornik1989multilayer}, there exists a neural function $\operatorname{NN}_r$ that approximates $f_r$ arbitrarily well and thus satisfies the condition in Lemma \ref{lemma:janossy-works}. 
\end{proof}

\end{document}